\documentclass{article}
\usepackage[utf8]{inputenc}
\usepackage[english]{babel}
\usepackage[dvipsnames]{xcolor,colortbl}
\usepackage{url}
\usepackage{graphicx}
\usepackage{amsfonts}
\usepackage{amsmath,amssymb}
\usepackage{amsthm}
\usepackage[linesnumbered, ruled]{algorithm2e}

\theoremstyle{definition}
\newtheorem{definition}{Definition}[section]

\newtheorem{theorem}{Theorem}[section]

\newtheorem{lemma}[theorem]{Lemma}
\theoremstyle{remark}
\newtheorem*{remark}{Remark}
\newtheorem*{example}{Example}

\DeclareGraphicsExtensions{.pdf,.jpg,.png}

\definecolor{mygreen}{rgb}{0,0.6,0}
\definecolor{mygray}{rgb}{0.5,0.5,0.5}
\definecolor{mymauve}{rgb}{0.58,0,0.82}

\def\result#1{\emph{\color{blue} #1}}

\newcommand{\tuple}[1]{\ensuremath{\left \langle #1 \right \rangle }}

\def\sideNote#1{\marginpar{\indent \color{gray}{\footnotesize{\emph{#1}}}}}

\newcommand{\silence}[1]{}

\newif\ifdraft

\def\inlineRem#1{{\noindent\bf\fcolorbox{red}{white}{\color{blue}{$\|$~{#1}~$\|$}}}}

\def\sticky#1{\colorbox{yellow}{\scriptsize\color{blue} @@} \marginpar{\scriptsize\color{red} #1}}

\newcommand{\rem}[1]{}

\newcommand{\diffRem}[1]{}

\title{A Cube Algebra with Comparative Operations: Containment, Overlap, Distance and Usability}
\author{Panos Vassiliadis\\
Dept. of Computer Science and Engineering\\
University of Ioannina\\
Ioannina 45110, Hellas\\
pvassil@cs.uoi.gr
}

\begin{document}

\maketitle

\begin{abstract}
In this paper, we provide a comprehensive rigorous modeling for multidimensional  spaces  with  hierarchically  structured  dimensions  in  several layers of abstractions and data cubes that live in such spaces. We model cube queries and their semantics and define typical OLAP operators like Selections, Roll-Up, Drill-Down, etc. The model serves as the basis to offer the main contribution of this paper which includes theorems and algorithms for being able to associate data cube queries via comparative operations that are evaluated only on the syntax of the queries involved. Specifically, these operations include: (a) foundational containment, referring to the coverage of common parts of the most detailed level of aggregation of the multidimensional space, (b/c) same-level containment and intersection, referring to the inclusion/existence of common parts of the multidimensional space in two query results of the same aggregation levels, (d) query distance, referring to being able to assess the similarity of two queries in the same multidimensional space, and, (e) cube usability, i.e., the possibility of computing a new cube from a previous one, defined at a different level of abstraction.
\end{abstract}

\section{Introduction}
Multidimensional spaces with hierarchically structured dimensions over several levels of abstraction, along with data cubes, (i.e., structured collections of data points at the same level of detail in the context of such spaces) -- all bundled under the focused term  On-Line analytical Processing (OLAP), or the broader encompassing term Business Intelligence tools -- provide a paradigm for data management whose simplicity is hard to match. Nowadays, there is a proliferation of data management paradigms, data science and analytics frameworks and tools, and a solid move from the business analyst's needs to the needs of the data scientist and the data journalist. In all these attempts, however, the notion of hierarchically structured dimensions is not actually being used as an inherent part of the modern data visualization tools, notebooks or similar approaches (the need for defining and populating the levels of the multidimensional space being a serious reason for that). Thus, despite the fact that OLAP technology, is practically 30 years old at the time of the writing of these lines, the idea of organizing data in \emph{simple} data cubes that can be manipulated at \emph{multiple levels of detail} has not been replaced or surpassed by the current trends in any way.

In this paper, we start on the assumption that there is merit is using the multidimensional paradigm as a basis for data management by end-users -- see \cite{DBLP:conf/dolap/VassiliadisM18, DBLP:journals/is/VassiliadisMR19} with its extension with information-rich new constructs to address the new requirements of the 2020's. \result{The problem that this paper addresses is, however, quite more fundamental and can be summarized as the definition of a cube algebra extended with comparative operators between cube queries, and in the introduction of theorems and algorithms for checking and performing these operators -- or in a more concise way: "how can we compare the results of two data cube queries?"}. The problem of defining the comparative operations between two cubes boils down to very traditional problems explored by the relational world decades ago, like: can we tell when one query is contained within another, or, which is the not-common part of two (cube) queries based only on their syntax? The particularity of hierarchical dimensions is that containment, intersection, distance and the rest of the comparative operations between constructs in the multidimensional spaces such dimensions form are not only the explicit ones, at the specific levels at which two data cubes are defined, but also implicit ones, at different levels of abstraction. To the best of our knowledge, an explicit framework for handling the comparative operations of two cubes does not exist to the day.

\textbf{Example and Motivation}. The main reason for addressing the problem can be discussed on the grounds of an example, as depicted in Figure~\ref{fig:ex}. Assume a multidimensional space involving data for the customers of a tax office over three dimensions, specifically \emph{Time}, \emph{Education} and \emph{Workclass}, for which \emph{taxes paid} and \emph{hours devoted} are recorded by the tax office. The dimensions are structured in layers of levels, and for clarity reasons, we label the levels $L1$, $L2$, etc., with the higher index for a level indicating a high level of coarseness ($L0$ being the most detailed level of all in a dimension, and $ALL$ being the most coarse one with a single value 'All'). The analysts of the tax office, at the end of a year, fire several analytical queries, to understand the behavior of the monitored population better. An automatic query generator/recommender tool aids the data exploration by automatically generating queries. Assume to queries, $Q1$ and $Q2$ have been issued already by an analyst, as depicted in Figure~\ref{fig:ex}, sharing the same groupers ($Month$, $Workclass.L1$ and $Education.L2$), but with different selection conditions. \emph{Assume the query recommender automatically generates the query $Q3$, depicted on top of the Figure~\ref{fig:ex}. Should it recommend it to the user?}
The decision can be evaluated on the grounds of different dimensions:
\begin{itemize}
    \item Is $Q3$ \underline{relevant}? Relevance refers to exploring parts of the multidimensional space that are in the focus of the user's interest. One interpretation of relevance is that the more visited a certain subset of the space, the more relevant it seems to be for the user. We can see in Figure~\ref{fig:ex}, that the area pertaining to the intersection of $Q1$ and $Q2$ (year 2019 that is) seems to be the most "hot", whereas the area of $Q3$ ls less "hot". So, we need mechanisms for query intersection between new and old queries.
    \item Is $Q3$ \underline{peculiar and diverse} with respect to the current session? As another alternative, one might want to answer the question: how "far" is the new query Q3 from the previous ones? (implying that the farthest a query is from the previous ones, the more novel it is). In this case, a way to infer the distance of two queries based on their syntax is also necessary.
    \item Is $Q3$ \underline{novel}? Novelty refers to the revelation of new facts to the user. Although the observant reader might suggest that since $Q3$ is at a different level of granularity than the previous queries, and thus, necessarily novel, at the same time, it is also true that in terms of the area the query covers in the multidimensional space, the area is a clear subset of the area of $Q2$. We thus need a mechanism to detect inclusion, overlap and non-overlap of queries at the detailed level, at least as a pre-requisite to facilitate the assessment of novelty. This includes both the query as a whole, but also which subset of its cells are (non)overlapping with previous queries. 
\end{itemize}
    
\begin{figure}[tb]
  \centering
    \includegraphics[width=0.9\textwidth]{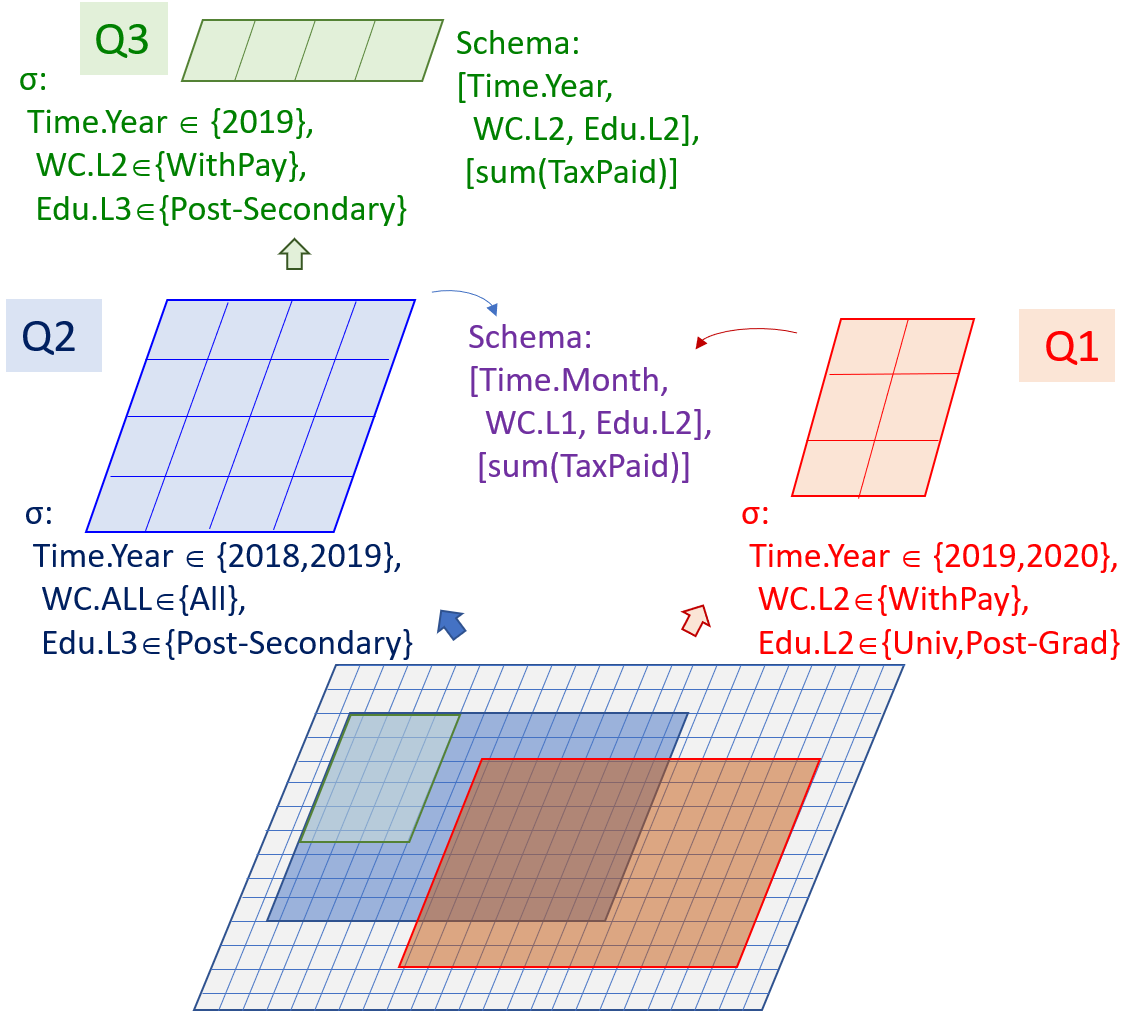}
      \caption{An example of 3 cube queries with different relationships to each other, all defined over the same data set, at the most low level of detail }
      \label{fig:ex}
\end{figure}

Being at different levels of granularity is only one aspect of the overlap of two queries: assume now that someone asks "which part of $Q2$ is novel with respect to $Q1$?". The problem that arises here has several characteristics: (a) the two cubes are at the same level of detail, for all their dimensions, and thus, the question is not necessarily needing an answer at the detailed level; (b) the question is not a Boolean one ("is the query novel?", but a fractional one: "which cells of $Q2$ can be enumerated in a 'novel' subset of the result of $Q2$?"); (c) even if two cells belonging to the two cubes have the same coordinates, how can we guarantee that they have been computed from the same detailed facts at the base level? For all these, we need mechanisms to appropriately enumerate the subset of two cube query results that is indeed common, without computing the results.

An extremely important requirement here is to be able to do this kind of computation without executing the queries, and thus having their cells at hand, but, deciding on these relationships based only on their syntactic definition. This is of uttermost importance as actually executing the queries and comparing the cells of the results is a lot more costly than comparing the syntax of the queries only.

Bear in mind, also, that most of the above problems can come in two variants: (a) an existential variant that is answered via a Boolean answer (e.g., "is the query $Q3$ novel with respect to the detailed area it covers?"), or (b) a quantitative variant that is answered with a score (e.g., "can we give a relevance score to $Q3$?").

\textbf{Contribution}. Traditionally, related work has handled the problem of query containment and view usability for the relational case (see Section 2 for a discussion, and \cite{Hale01}, \cite{Cohe05}, \cite{DBLP:reference/db/Cohen09} and \cite{DBLP:reference/db/Vassalos09} as reference pointers to the related work). The existence of hierarchically structured dimensions in the case of multidimensional spaces with different possible levels of aggregations, as a context for the determination of cube usability has not been extensively dealt with by the database community, however. \result{We attempt to fill this gap by providing a comprehensive rigorous modeling and the respective theorems and algorithms for being able to associate data cube queries via comparative operations}.

The contributions of this paper can be listed as follows.
\begin{itemize}
    \item In Section~\ref{sec:bckgr}, we provide a comprehensive model for multidimensional hierarchical spaces, cubes and cube queries, that can facilitate a rich query language. We categorize selection conditions with respect to their complexity. Moreover, we also show how the most typical operations in OLAP, like Roll-Up, Drill-Down, Slice and Drill-Across can all be modeled as queries of our model.
    \item Based on the intrinsic property of the model that all query semantics are defined with respect to the most detailed level of aggregation in the hierarchical space, and in contrast to all previous models of multidimensional hierarchical spaces, in Section~\ref{sec:DescSign}, we accompany the proposed model with definitions of equivalent expressions at different levels of granularity. We introduce the necessary terminology and notation, too, to solidify these concepts in the vocabulary of multidimensional modeling. Specifically, we introduce (a) proxies, i.e., equivalent expressions at different levels of abstraction, (b) signatures, i.e., sets of coordinates specifying a "border" in the multidimensional space that specifies a sub-space pertaining to a model's construct, and, (c) areas, i.e., set of cells enclosed within a signature.
    \item We introduce the problem of whether a cube $c^n$ is fundamentally contained within another cube $c^b$, i.e., whether the area at the lowest level of aggregation that pertains to it is a subset of the respective area that pertains to $c^b$, in Section~\ref{sec:foundCont}. We introduce theorems for both the decision (i.e., Boolean) and the enumeration variant (i.e., which cells are outside the jointly covered area) of the problem.
    \item Having introduced the containment problem at the most detailed level, in Section \ref{sec:SameLevelCont}, we move on to tackle the problem of containment for cubes sharing the exact same schema, which we call the same-level containment problem, under the constraint of not computing the result of the queries, but using only their syntactic expression. To address the problem, which comes with the complexity of having to deal with grouper dimensions where selections have also been posed, we introduce the notion of rollability which refers to the property of the combination of a filter and a grouper level at the same dimension to produce result coordinates that are fully covering the respective subspace at the most detailed level. Then, we introduce the decision and the enumeration problem and provide checks and algorithms for both of them.
    \item In Section~\ref{sec:intersect}, we deal with the problem of (syntactic) query intersection, i.e., deciding whether, and to what extent, the results of two queries overlap, given their syntactic expression only. Again, we introduce the decision and enumeration problems, as well as the enumeration problem of a query being tested for intersection with the members of a query set.
    \item Departing from containment and intersection problems, in Section~\ref{sec:queryDistance}, we enrich the methods discussed by the paper by introducing a set of formulae for the evaluation of the distance of two queries.
    \item Usability: In Section~\ref{sec:cubeUsa}, we discuss the possibility of computing a new cube from a previous one, defined at a different level of abstraction; we introduce the respective test as well as a rewriting algorithm. 
    
\end{itemize}

Section \ref{sec:rw} provides a discussion of the related work. 
Section \ref{sec:concl} provides conclusions and open issues for future work.
\newpage

\section{Related Work}\label{sec:rw} 

\subsection{Models for hierarchical multidimensional spaces and cubes }
There is an abundance of \emph{models for multidimensional hierarchical databases, cubes, and OLAP} that is very well surveyed in \cite{DBLP:conf/dawak/RomeroA07}. \cite{DBLP:conf/dawak/RomeroA07} evaluates a number of formal models on the support of operations like: 
(a) Set operations like union, difference and intersection between cubes,
(b) Selection, i.e., the application of a filter over a cube, 
(c) Projection, i.e., the filtering-out of some unwanted measures from a cube, 
(d) Drill-Across, by joining a new measure to the cube (hiding possibly the join of different fact tables, at the physical layer),
(e) Roll-Up and (f) Drill-Down by changing the coarseness of the aggregation to more coarse, or finer levels of detail, respectively, 
(g) ChangeBase by re-ordering the sequence of levels in the schema of a cube (which is mostly a presentational, rather than a logical operation -- e.g., to perform a pivot), or assigning the same cells to a different multidimensional space.

We refer the reader to \cite{DBLP:conf/dawak/RomeroA07} for the details of the different models that have been proposed in the literature. In this paper, we extend a more restricted version of the model \cite{DBLP:journals/is/VassiliadisMR19} and formalize rigorously the entire domain of  multidimensional spaces with hierarchical dimensions, data cubes and cube queries, show our support for the most fundamental operations and use this model as the basis for the rest of the material concerning the comparison operations between cubes, like containment, intersection, etc. We refer the reader to Section~\ref{sec:paperNovelty} for a discussion of the differences with \cite{DBLP:journals/is/VassiliadisMR19} with respect to the modeling part (as well as for the completely new parts on comparative operations).

\subsection{Relationships between views: usability and containment}
The literature around the answering of queries via previously
answered query results (i.e., views or cached queries) is
typically organized around three main themes, which we present in
an order of increasing difficulty :
\begin{itemize}
  \item the \emph{query containment problem} is a decision problem where the goal is to determine whether the set of
  tuples computed by a query $Q$ is always a subset of the set of
  tuples  produced by a query $Q'$ independently of the contents
  of the database over which both queries are defined
  \item the \emph{view usability problem} is a similar problem that concerns the decision on whether a query $Q$
defined over a set of relations $S^{Q}$ \emph{can be} answered via
a view $V$ (and possibly a set of auxiliary relations $S^{V}$
$\subseteq$ $S^{Q}$) such that the resulting set of tuples is
identical, independently of the contents of the underlying
database
  \item the \emph{query rewriting problem} concerns \emph{how} the
query $Q$ must be rewritten in order to be answered via $V$
\end{itemize}

Excellent surveys and lemmas exist that summarize these areas. Halevy \cite{Hale01} addresses the
general problem of answering queries using views.  A dedicated survey of the special problem
of aggregate query containment by Sara Cohen is \cite{Cohe05}. Two lemmas on the topic are \cite{DBLP:reference/db/Cohen09} and \cite{DBLP:reference/db/Vassalos09}.
We refer the
interested reader to all the aforementioned surveys and lemmas for a broader coverage of the topic.


\subsubsection{Query Containment}
As typically happens in the database literature, simple
conjunctive queries are the basis of the research efforts in the
area of view usability. The problem of query containment for conjunctive queries (without any form of aggregation) has been extensively studied since 1977, when Chandra and Merlin \cite{CM77} provided their famous results on the NP-completeness of finding homomorphisms between conjunctive queries. We refer the interested reader to \cite{DBLP:conf/pods/NuttSS98} and \cite{DBLP:journals/tods/CohenNS06} for a discussion of a large set of papers dealing with different aspects of conjunctive query containment.

In \cite{DBLP:conf/pods/NuttSS98} (and its long version \cite{DBLP:journals/jacm/CohenNS07}), Nutt et al., are concerned with the problem of equivalences between aggregate queries. The paper explores the case of conjunctive queries with simple selection conditions and typical aggregate functions. The selection conditions involve simple comparisons between attributes or attributes and values. The aggregate functions involve $min$, $max$, $sum$, $count$ and $count~distinct$. Assuming two queries $q$ and $q'$, the goal of the approach is (a) to check whether the heads of the queries are compatible (in other words, the grouping attributes and the aggregate function are compatible), and, (b) to find homomorphisms between the variables of $q$ and $q'$. Due to the problem of the $distinct~count$ the paper discriminates between set and bag semantics for its aggregate queries. In \cite{CoNS99}, Cohen et al., extend the results of \cite{DBLP:conf/pods/NuttSS98} by handling disjunctive selections too.
However, in all the above works, the usage of multidimensional data with dimensions including hierarchies is absent: all the methods operate on simple relational data and Datalog queries extended with aggregations.

\subsubsection{View usability}
Introducing views into the definition of a query in order to replace 
some relations of the original query definition is not as straightforward
as one would typically expect. Following \cite{Hale01}, we can informally
say that in the simple case where both the view $V$ and the query
$Q$ are conjunctive queries, view usability requires that there is
a mapping of the relations involved in the view to the relations
involved in the query; the view must provide looser selection
conditions than the query; and, finally, the view must contain all
the necessary fields that are needed in order to (a) apply all the
necessary extra selection conditions to compensate for the looser
selection of the view and (b) retrieve the final result of the
query (practically the SELECT clause of the query). When aggregate
views are involved, the situation becomes more complicated, since
we must guarantee that the "conjunctive part" of the view and the
query must produce the same "base" over which equivalent
aggregations are performed (keep in mind that aggregations have
the inherent difficulty of having to deal with the problem of
producing the same number of tuples correctly -- a.k.a. the
notorious 'count' problem). 

Larson and Yang in \cite{DBLP:conf/vldb/LarsonY85} provide a solution to the problem of view usability for views and queries that are simple Select-Project-Join (SPJ) queries.

Das et al \cite{DJLS96} have provided
a paper handling view usability for a large number of SQL query
classes. A particular feature of the paper is that the method
handles both the case where the view is not an aggregate view and
the case where the view is also performing an aggregation over the
underlying data. The paper is also accompanied by algorithms to rewrite the queries over the views.

\subsubsection{Query rewriting}
The rewriting problem has received attention from both a theoretical and a practical perspective; the former deals with theoretical establishment of equivalences whereas the second follows an optimizer-oriented approach.

Levy et al in \cite{DBLP:conf/pods/LevyMSS95} provide a first simple algorithm for rewriting conjunctive queries. This is done by determining the relations that can be removed from a query once a view is used (instead of them). The paper investigates also the cases of $minimal$ and $complete$ rewritings. Minimal rewritings are the ones where literals cannot be further removed and complete rewritings are the ones where only views participate in the new query.

Chaudhuri et al in \cite{DBLP:conf/icde/ChaudhuriKPS95} deal with the optimization of SPJ queries (without aggregations) by extending the join enumeration part of a traditional System-R optimizer with the possibility of considering materialized SPJ views, too. In a similar fashion, Gupta et al in \cite{DBLP:conf/vldb/GuptaHQ95} consider the problem for the case of aggregate queries and views by introducing generalized projections as part of the query plan and pulling them upwards or pushing them downwards in it. Similarly, Chaudhuri and Shim \cite{DBLP:conf/edbt/ChaudhuriS96} explore the problem of pulling up or pushing down aggregations in a query tree.

Returning back to the theoretical perspective, Cohen et al \cite{CoNS99} deal with the case of aggregate query
rewriting for the cases where the aggregate function is sum or
count. Specifically, the paper deals with the problem of replacing 
a query $q$ defined over the database $D$ with a new, equivalent 
query $q'$ that also includes views from a set $V$. The main idea of 
the paper is to unfold the definitions of the views, so that the comparison 
is done in terms of query containment. 

Grumbach and Tininini \cite{GT03} explore the
rewriting problem for aggregate views for the case of views and
queries without any selection conditions at all. In \cite{DBLP:journals/acta/GrumbachRT04} the authors introduce a new syntactic equivalence relation between conjunctive queries, called isomorphism modulo a product to capture the multiplicity of duplicates, and discuss the problem of obtaining isomorphisms, which proves to be NP-complete. The authors discuss the view usability problem for bag-views.

In \cite{DBLP:conf/icdt/CohenNS03} the authors provide containment characterizations and rewritings for count queries. In \cite{DBLP:journals/tods/CohenNS06}, the authors provide a discussion of aggregation functions, the identification of rewriting candidates and the problem of existence of a rewriting.

\subsubsection{Multidimensional hierarchical space of data}
All the aforementioned approaches work with plain relational data, whose attributes are plain relational attributes. What happens though when we need to work in multidimensional \emph{hierarchical} spaces, i.e., with dimensions involving hierarchies?

Concerning the OLAP field, the case of cube usability resolves in queries and views having joins between a fact table and its dimension tables. However, although implications between an atom of the view and an atom of the query can be handled if they are
defined over the same attribute, to the best of our knowledge, the only work where it is possible to handle the implications between atoms defined at different levels (i.e., attributes) is \cite{DBLP:conf/caise/VassiliadisS00} (long v., at \cite{VassiliadisPhdthesis}). This has to do both with the case where the atoms are of the form $L$ $\theta$ $value$ (along with the marginal constraints for the values involved) and with the case where the atoms are of the form $L$ $\theta$ $L'$ (where implications among different levels have to be defined via a principled reasoning mechanism), with $\theta \in \{=, <, >, \leq, \geq \}$. 

Theodoratos and Sellis \cite{DBLP:conf/ssdbm/TheodoratosS00} also propose a reasoner-based approach. To the best of our understanding, the mechanism of performing the reasoning between different levels is unclear. Moreover, the handling of the combination of selections and aggregation is not explicit, thus requiring additional constraints for the method to work.  

\subsection{Comparative Discussion}\label{sec:paperNovelty}

Overall we would like to stress that our approach is one of the first attempts to comprehensively introduce comparative operations that are particularly tailored for the context of hierarchical multidimensional data (i.e., in the presence of hierarchies) and cube (i.e., query) expressions defined over them. Specifically, compared to previous works, the current papers produces the following novel aspects:
\begin{enumerate}

\item The paper comes with a comprehensive model for hierarchical multidimensional spaces and query expressions in them. Compared to the multidimensional model of \cite{DBLP:journals/is/VassiliadisMR19}, we provide the following extensions:
\begin{itemize}
    \item We slightly improve notation.
    \item We discriminate different classes of selection conditions and discuss proxies, areas and signatures.
    \item We extend the working type of selection atoms to set-valued atoms (which practically poses all the subsequent issues on a new basis). All the following contributions are completely novel with respect to \cite{DBLP:journals/is/VassiliadisMR19}.
\end{itemize}
\item The paper comes with a principled set of tests for the decision problem of testing containment at various level of detail (specifically: foundational, same-level, and, different level containment), query intersection (at various levels of detail), and query distance as well as for the enumeration problem of reporting on the specific cells that fall within/exceed the boundaries of containment and intersection.
\item Compared to previous work on query rewriting for hierarchical multidimensional spaces, we work with set-valued rather than single-valued selections (although we do not cover comparisons between levels, or with arbitrary comparators other than equality).

\end{enumerate}

\newpage
\section{Formalizing data, dimension hierarchies cubes and cube queries}\label{sec:bckgr}

In this Section, we give the formal background of our modeling concerning multidimensional databases, hierarchies and queries. 



As typically happens with multidimensional models, we assume that \emph{dimensions} provide a \emph{context} for facts
\cite{DBLP:series/synthesis/2010Jensen}. This is especially
important considering that dimension values come in
\emph{hierarchies}; every single fact can be
simultaneously placed in multiple hierarchically-structured
contexts, thus giving users the possibility of analyzing sets of facts from
different perspectives. The underlying data sets include \emph{measures} that are characterized with respect to these dimensions. \emph{Cube queries} involve measure aggregations at specific levels of granularity per dimension, along with filtering of data for specific values of interest. 



\subsection{Domains, dimensions and underlying data}

\textbf{Domains}. We assume the following infinitely countable and
pairwise disjoint sets: a set of \emph{level names} (or simply
\result{levels}) $\mathcal{U_{L}}$, a set of \emph{measure names}
(or simply \result{measures}) $\mathcal{U_{M}}$, a set of
\emph{regular~data~columns} $\mathcal{U_{A}}$, a set of
\emph{dimension} names (or simply \emph{dimensions})
$\mathcal{U_{D}}$ and a set of \emph{cube} names (or simply
\emph{cubes}) $\mathcal{U_{C}}$. The set of \emph{data~columns}
 $\mathcal{U}$ is defined as $\mathcal{U}$ =
$\mathcal{U_{L}}$  $\cup$ $\mathcal{U_{M}}$ $\cup$
$\mathcal{U_{A}}$. For each $L$ $\in$ $\mathcal{U_{L}}$, we define
a countable totally ordered set $dom(L)$, the domain of $L$, which
is isomorphic to the integers. Similarly, for each $M$ $\in$
$\mathcal{U_{M}}$, we define an infinite set $dom(M)$, the domain
of $M$, which is isomorphic either to the real numbers or to the
integers. The domain for the regular data columns of
$\mathcal{U_{A}}$ is defined in a similar fashion to the one of
measures. We can impose the usual comparison operators to all the
values participating to totally ordered domains ${\{}<,
>, \le, \ge {\}}$.\\

\textbf{Dimensions and levels}.A \result{dimension} $D$ is a lattice ($\mathbf{L}$,$\preceq$) such that:
\begin{itemize}
	\item $\mathbf{L}$ = $\{$$L_{1}$,$\ldots$ ,$L_{n}$$\}$, is a finite subset of $\mathcal{U_{L}}$.
	\item $dom(L_{i})$ $\cap$ $dom(L_{j})$= $\emptyset$ for every $i$ $\ne$ $j$.
	
	\item $\preceq$ is a non-strict partial order defined among the levels of $\mathbf{L}$.\diffRem{symbol update} 
	
	\item With $D$ being a lattice, it follows that there is a highest and a lowest level in the hierarchy. The highest level of the hierarchy is the level $D$.$ALL$ with a domain of a single value, namely '$D$.$all$', for which it holds that $L$ $\preceq$ $ALL$ for all other levels $L$ in $\mathbf{L}$. Moreover, there is also the \emph{lowest~level} in the dimension, $D.L^0$, for which it holds that $L^0$ $\preceq$ $L$ for all other levels $L$ in $\mathbf{L}$. Whenever two levels are related via the partial order, say $L_{low}$ $\preceq$ $L_{high}$, we refer to $L_{low}$ as the descendant and to $L_{high}$ as the ancestor. 
\end{itemize}

Each path in the dimension lattice, beginning from its upper bound and ending in its lower bound is called a \emph{dimension path}. The values that belong to the domains of the levels are called \result{dimension members}, or simply \emph{members} (e.g., the values $Paris$, $Rome$, $Athens$ are members of the domain of level $City$, and, subsequently, of dimension $Geography$).\rem{must sync if definition of level is updated with a different treatment of properties}

\hrulefill

\begin{remark}
The reader is reminded, that a \textit{non-strict partial order} is reflexive (i.e., $L$ $\preceq$ $L$), antisymmetric (i.e., $L_1$ $\preceq$ $L_2$ and $L_2$ $\preceq$ $L_1$ means that $L_1$ = $L_2$), and transitive (i.e., $L_1$ $\preceq$ $L_2$ $\preceq$ $L_3$ means that $L_1$ $\preceq$ $L_3$).  As usually, we have to make two orthogonal choices: (a) whether the order is partial or total, and (b) whether the order is strict or non-strict. 
\begin{itemize}
	\item A partial order differs from a \textit{total order}, or \textit{chain}, in the part that it is possible that two elements of the domain can be non-comparable in the former, but not in the latter. Thus, we can have lattices, like the one for time, where $Week$ is not comparable to $Month$, but both precede $Year$ and follow $Day$. When we have to deal with chains (practically: instead of a lattice, the dimension levels form a linear chain), we will explicitly say so. 
    \item The intuition behind a non-strict order is "not higher than". A strict order, frequently denoted via the symbol $\prec$ on the other hand, revokes the reflexive property (i.e., $L$ $\nprec$ $L$, with the meaning "lower than"). The rationale for choosing a non-strict order, instead of a strict one, is convenience in the uniformity of notation. As we shall see later, we will introduce the notation $anc_{L_1}^{L_2}()$, and, we would like to be able to use the notation $anc_{L}^{L}()$ (effectively meaning $L$), without having to treat it as a special case.
\end{itemize}

\end{remark}

\hrulefill

To ensure the consistency of the hierarchies, a family of \textit{ancestor} functions $anc_{L_{1}}^{L_{2}}$ is defined, satisfying the following conditions:\sideNote{Constraints for ancestors and descendants}
\begin{enumerate}
	\item For each pair of levels $L_{1}$ and $L_{2}$ such that $L_{1}$ $\preceq$ $L_{2}$,
	the function $anc_{L_{1}}^{L_{2}}$ maps each element of $dom(L_{1})$ to an element of
	$dom(L_{2})$.
	
	\item Given levels $L_{1}$, $L_{2}$ and $L_{3}$ such that $L_{1}$ $\preceq$ $L_{2}$ $\preceq$
	$L_{3}$, the function $anc_{L_{1}}^{L_{3}}$ equals to the composition $anc_{L_{1}}^{L_{2}}$ $\circ$ $anc_{L_{2}}^{L_{3}}$.
	This implies that:
	\begin{itemize}
		\item $anc_{L_{1}}^{L_{1}}(x)$ = $x$.
		
		\item if $y$ = $anc_{L_{1}}^{L_{2}}(x)$ and $z$ = $anc_{L_{2}}^{L_{3}}(y)$, then $z$ = $anc_{L_{1}}^{L_{3}}(x)$.
		
		\item for each pair of levels $L_1$ and $L_2$ such that $L_1$ $\preceq$ $L_2$, the function $anc_{L_{1}}^{L_{2}}$ is monotone (preserves the ordering of values). In other words:
		
		$\forall$ $x$,$y$ $\in$ $dom(L_{1})$: $x$ $<$ $y$ $\Rightarrow$ $anc_{L_{1}}^{L_{2}}(x)$ $\le$
		$anc_{L_{1}}^{L_{2}}(y)$, $L_{1}$ $\preceq$ $L_{2}$
	\end{itemize}
	
	\item
	For each pair of levels $L_1$ and $L_2$ such that $L_1\preceq L_2$ the
	$anc_{L_{1}}^{L_{2}}$ function determines a set of finite equivalence classes $X_i$
	such that:
	\[(\forall x,y \in dom(L_1))\
	(anc_{L_{1}}^{L_{2}}(x)= anc_{L_{1}}^{L_{2}}(y)
	\Rightarrow x~and~y \mbox{ belong to the same } X_i).\]
	
	\item
	The relation $desc^{L_{low}}_{L_{high}}$ is the inverse of the $anc_{L_{low}}^{L_{high}}$
	function, i.e., 
	\[desc^{L_{low}}_{L_{high}}(v_h)= \{v_l \in dom(L_{low}):anc_{L_{low}}^{L_{high}}(v_l)=v_h\}.\]
\end{enumerate}

%

Observe that $desc(\cdot)$ is not a function, but a relation. With $anc(\cdot)$ and $desc(\cdot)$ we can compute the corresponding values of a dimension path at different levels of granularity in o(1).

\textbf{Level properties}. Levels can also also annotated with properties. For each level $L$, we define a finite set of functions, which we call properties, that annotate the members of the level. So, for each level $L$, we define a finite set of functions $\mathcal{F}^L$ = $\{F^L_1, ~\ldots,~F^L_k\}$, with each such function $F^L_i$ mapping the domain of $L$ to a regular data column $A_i$, s.t., $A_i$ $\in$  $\mathcal{U_{A}}$, i.e., $F^L_i$: $dom(L)$ $\rightarrow$ $dom(A_i)$.
So, for example, for the level $City$, we can define the functions $population()$ and $area()$. Then, for the value $Paris$ of the the level $City$, one can obtain the value $2M$ for  $population(Paris)$ and $100Km^2$ for $area(Paris)$.\\

\textbf{Schemata}. First, we define what a \textbf{schema} is in a multidimensional space.\\

A \emph{schema} $\mathbf{S}$ is a finite subset of
$\mathcal{U}$.\\

A \emph{multidimensional schema} is divided in two parts:
$\mathbf{S}$ = [$D_{1}$.$L_{1}$, $\ldots$, $D_{n}$.$L_{n}$,
$M_{1}$, $\ldots$, $M_{m}$], where:
\begin{itemize}
	\item $\{$$L_{1}$,$\ldots$ ,$L_{n}$$\}$ are levels from a dimension set
	$\mathbf{D}$ = $\{$$D_{1}$,$\ldots$, $D_{n}$$\}$ and level $L_{i}$ comes from dimension
	$D_{i}$, for 1 $\le$ $i$ $\le$ $n$.
	
	\item  $\{$$M_{1}$,$\ldots$, $M_{m}$$\}$  are measures.
\end{itemize}

A \emph{detailed multidimensional schema} $\mathbf{S}^{0}$ is a
schema whose levels are the lowest in the respective dimensions.\\

\textbf{Facts and cubes}. Now we are ready to define what a fact is, expressed as a \textbf{cell}, or multidimensional tuple  in the multidimensional space.\\

A \emph{tuple} under a schema $\mathbf{S}$ = [$A_{1}$, $\ldots$,
$A_{n}$] is a point in the space formed by the Cartesian Product 
of the domains of the attributes $A_i$,
$dom(A_{1})$ $\times$ $\ldots$ $\times$ $dom(A_{n})$, such that
$t[A]$ $\in$ $dom(A)$ for each $A$ $\in$ $\mathbf{S}$.\\


A \emph{multidimensional tuple}, or equivalently, a \result{cell} or a \emph{fact}, $t$ is a tuple under a multidimensional schema $\mathbf{S}$ =
[$D_{1}$.$L_{1}$, $\ldots$, $D_{n}$.$L_{n}$, $M_{1}$, $\ldots$,
$M_{m}$].\\


Having expressed what individual pieces of data, or facts, are, we are now ready to define \textbf{data sets} and \textbf{cubes}.\\

A \emph{data set} $\mathbf{DS}$ under a schema $\mathbf{S}$ =
[$A_{1}$, $\ldots$, $A_{n}$] is a finite set of tuples under
$\mathbf{S}$.\\

A \emph{multidimensional data set} $\mathbf{DS}$, also referred to as a \result{cube}, under a schema
$\mathbf{S}$ = [$D_{1}$.$L_{1}$, $\ldots$, $D_{n}$.$L_{n}$,
$M_{1}$, $\ldots$, $M_{m}$] is a finite set of cells under
$\mathbf{S}$ such that:

\begin{itemize}
	\item $\forall$ $t_{1}$, $t_{2}$ $\in$ $\mathbf{DS}$, $t_{1}$[$L_{1}$,$\ldots$, $L_{n}$] = $t_{2}$[$L_{1}$, $\ldots$, $L_{n}$] $\Rightarrow$ $t_{1}$ = $t_{2}$.
	\item for no strict subset $X$ $\subset$ ${\{}L_{1},\ldots ,L_{n}{\}}$, the previous also holds.
\end{itemize}


\noindent In other words, $M_{1}$, $\ldots$, $M_{m}$ are functionally
dependent (in the relational sense) on levels
$\{$$L_{1}$,$\ldots$ ,$L_{n}$$\}$ of schema $\mathbf{S}$.
Notation-wise, we use the expression $c$ $\in$ $\mathbf{DS}$ when a cell belongs to a multidimensional data set, and the expression $\mathbf{DS}.cells$ to refer to the set of tuples of a multidimensional data set.\\

A \emph{detailed multidimensional data set} $\mathbf{DS}^{0}$, also referred to as a \result{basic cube}, is a data set under a detailed schema $\mathbf{S}^{0}$.\\

A \emph{star schema} ($\mathbf{D}$,$\mathbf{S}^{0})$ is a couple
comprising a finite set of dimensions $\mathbf{D}$ and a detailed
multidimensional schema $\mathbf{S}^{0}$ defined over (a subset
of) these dimensions.

\subsection{Selections}\diffRem{Significant updates}

\textbf{Selection filters}. An \emph{atom} is an expression that takes one of the following forms: 
\begin{itemize}
    \item a Boolean value, i.e., $true$ or $false$ (with obvious semantics),
    \item $anc_{L_{0}}^{L}(L)$ $\theta$ $v$, or in shorthand, $L$ $\theta$ $v$,
with $v$ $\in$ $dom(L)$ and $\theta$ is an operator from the set $\{ >, <, =, \ge, \le, \ne \}$; equivalently, this expression can also be written as $L$ $\theta$ $v$.
    \item $anc_{L_{0}}^{L}(L)$ $\in$ $V$, $V$ being a finite set of values, $V$ = $\{ v_1, \ldots, v_k \}$, $v_i$ $\in$ $dom(L)$; equivalently, this expression can also be written as $L$ $\in$ $V$.
\end{itemize}

A \emph{conjunctive expression} is a finite set of atoms connected via the logical connectives $\wedge$.\\

A \emph{selection condition} $\phi$ is a formula involving atoms and the logical connectives $\wedge$, $\vee$ and $\neg$. The following subclasses of selection conditions are of interest:
\begin{itemize}
    \item A \emph{selection condition in disjunctive normal form} is a selection condition connecting conjunctive expressions via the logical connective $\vee$. 
    \item A \emph{multidimensional conjunctive selection condition} $\phi$ applied over a multidimensional data set $\mathbf{DS}$ is a selection condition with the following constraints: (a) it involves a single composite conjunctive expression, and, (b) there is exactly one atom per dimension of the schema of $\mathbf{DS}$.
    \item A \emph{dicing selection condition} $\phi$ applied over a multidimensional data set $\mathbf{DS}$ is a multidimensional conjunctive selection condition whose atoms are all of the form $anc_{L_{0}}^{L}(L)$ = $v$, or in shorthand, $L$ = $v$, (this also includes the special case of $anc_{L_{0}}^{ALL}(ALL)$ = $all$ for dimensions that would otherwise come with a $true$ atom). The term `dicing' is a typical term for such selection conditions in the OLAP domain.
    \item A \emph{simple selection condition} $\phi$ applied over a multidimensional data set $\mathbf{DS}$ is a multidimensional conjunctive selection condition whose atoms are all of the form $anc_{L_{0}}^{L}(L)$ $\in$ $V$ (this also includes the special cases of (a) a single-member set $V$, when an atom is of the form $L$ = $v$, and, (b) $anc_{L_{0}}^{ALL}(ALL)$ $\in$ $\{all\}$ for dimensions that would otherwise come with a $true$ atom).
\end{itemize}

The intuition behind the introduction of the above classes of selection conditions is that queries (see next) with selection conditions in disjunctive normal form can be handled as unions of queries with conjunctive expressions as selection conditions. Multidimensional selection conditions  come with a requirement for an atom per dimension, which is a convenience that will allow the homogeneous treatment of all dimensions in the sequel. Remember that $true$ is also an atom; practically equivalently, $D.ALL$ = $D.all$ includes the entire domain of a dimension's members. Thus, requiring an atom per dimension is easily achievable. \result{In the rest of our deliberations, wherever not explicitly mentioned for a certain dimension, an atom $D.ALL$ = $D.all$ is assumed}. \\

The \textit{semantics} of a selection condition is as follows: the expression $\phi(\mathbf{DS})$ produces a set of tuples $\mathbf{X}$ belonging to $\mathbf{DS}$ such that when, for all the occurrences of level names in $\phi$, we substitute the respective level values of every $x$ $\in$ $\mathbf{X}$, the formula $\phi$ becomes true.\\

A \emph{well-formed selection condition} is defined as a selection condition that is applied to a data set with all the level names that occur in it belonging to the schema of the data set and all the values of an atom pertaining to the domain of the respective level. \result{In the rest of our deliberations, unless specifically mentioned otherwise, we assume that all the selection conditions are simple, well-formed selection conditions}.\\

A \emph{detailed selection condition} $\phi^{0}$ is a selection condition where all participating levels are the detailed levels of their dimensions.\\


A multidimensional conjunctive selection condition $\phi$ produces an equivalent detailed selection condition, $\phi^0$, via the following mapping of atoms of $\phi$ to atoms of $\phi^0$ (remember that there is a single atom per dimension). \rem{examples of variants of $\phi$?} 
\begin{enumerate}
    \item Boolean atoms of $\phi$ are mapped to themselves in $\phi^0$
    \item Atoms of the form $anc_{L_0}^{L}(L)$ $\theta$ $v$, or in shorthand, $L$ $\theta$ $v$, are mapped to their detailed equivalents as follows, by exploiting the $desc$ mapping and the order-preserving monotonicity of the domains of all the levels: 
    \begin{itemize}
        \item $L$ = $v$ is mapped to $L_0$ $\in$ $desc_{L}^{L_0}(v)$
        \item $L$ $\ne$ $v$ is mapped to $L_0$ $\ne$ $desc_{L}^{L_0}(v)$
        \item $L$ $<$ $v$ is mapped to $L_0$ $<$ $min$($desc_{L}^{L_0}(v)$)
        \item $L$ $\le$ $v$ is mapped to $L_0$ $\le$ $max$($desc_{L}^{L_0}(v)$)
        \item $L$ $>$ $v$ is mapped to $L_0$ $>$ $max$($desc_{L}^{L_0}(v)$)
        \item $L$ $\ge$ $v$ is mapped to $L_0$ $\ge$ $min$($desc_{L}^{L_0}(v)$)
    \end{itemize}
    \item Atoms of the form $anc_{L_{0}}^{L}(L)$ $\in$ $V$, $V$ = $\{ v_1, \ldots, v_k \}$, $v_i$ $\in$ $dom(L)$ are mapped to  $anc_{L_{0}}^{L_{0}}(L_0)$ $\in$ $U$, or in shorthand $L_0$ $\in$ $U$, with $U$ = $\bigcup\limits_{i=1}^{k} desc_{L}^{L_0}(v_i)$    
\end{enumerate}

\hrulefill

\begin{remark}
Clearly, the above transformations require (and take advantage of) the monotonicity of the domains of the levels within a hierarchy (the second property of the ancestor family of functions). To forestall any possible criticism, here we discuss the feasibility, importance and consequences of this property.

First of all, \emph{feasibility}. With the exception of time-related dimensions, the vast majority of dimension levels are of nominal nature. For the time-related dimensions, the monotonicity property is inherent and not further elaborated. The nominal levels come with a finite set of discrete values, that do not necessarily hide any ordering, or any other isomorphism to the integers. Practically, these levels are internally represented via attributes in a Dimension table, and identified by Surrogate Keys, i.e., artificially generated integers, that allow the sorting of the values (although without any intuition of the sorting per se).  So, sorting and in fact, sorting with a respect of monotonicity between levels is feasible.

Second, \emph{importance}. The presence of a total ordering of the values, facilitates the direct rewriting of the expressions concerning high level intervals, to expressions also involving intervals at lower levels. So, any interval queries can be immediately translated to selection conditions at the most detailed level. Other than this, the model can work without the monotonicity property anyway. Also, in the absence of the monotone ordering, higher-level intervals can be translated to expressions involving set participation for the case of finite domains of dimensions (as typically happens in dimension tables).
So overall: monotonicity is feasible, useful for fast rewritings of range queries and its absence is amendable.

\end{remark}
\hrulefill

\subsection{Cube Queries and Sessions}
\textbf{Cube queries}. The user can submit \emph{cube queries} to
the system. A cube query specifies (a) the detailed data set over which
it is imposed, (b) the selection condition that isolates the
records that qualify for further processing, (c) the aggregator
levels, that determine the level of coarseness for the result, and
(d) an aggregation over the measures of the underlying cube
that accompanies the aggregator levels in the final result. More
formally, a \result{cube query}, is an expression of the form:


\[q = \tuple{\ \mathbf{DS}^{0},\ \phi,\ [L_1,\ldots,L_n,M_1,\ldots,M_m],\
[agg_1(M^0_1),\ldots,agg_m(M^0_m)]\ }
\]
where
\begin{enumerate}
	\item
	$\mathbf{DS}^{0}$ is a detailed data set over the schema $\mathbf{S}$ =[$L_1^0$,
	$\ldots$, $L_n^0$, $M_1^0$, $\ldots$ ,$M_k^0$], $m$ $\le$ $k$.
	\item
	$\phi$ is a multidimensional conjunctive selection condition,\rem{not necessarily dicing or simple} \diffRem{used to be $\phi_{0}$}
	\item
	$L_1,\ldots,L_n$ are \textit{grouper} levels such that $L^0_i\preceq L_i$, $1\le i\le n$,
	\item
	$M_1,\ldots,M_m$, $m\le k$, are aggregated measures (without loss of generality we assume that aggregation takes place over the first $m$ measures -- easily achievable by rearranging the order of the measures in the schema),
	\item
	$agg_1,\ldots,agg_m$ are aggregate functions from the set
	$\{sum,$ $min,max,count, avg\}$.
\end{enumerate}

The semantics of a cube query in terms of SQL over a star schema are:

\vspace{6pt}\begin{tabular}{l}
	\small\texttt{SELECT $L_{1}$,\ldots ,$L_{n}$, $agg_{1}(M_1^0)$ AS $M_1$,\ldots,$agg_{1}(M_m^0)$  AS $M_m$}\\
	\small\texttt{FROM DS$^{0}$ NATURAL JOIN D$_{1}$ \ldots\ NATURAL JOIN D$_{n}$}\\
	\small\texttt{WHERE $\phi^{0}$}\\
	\small\texttt{GROUP BY $L_{1}$,\ldots ,$L_{n}$}\\
\end{tabular}

\vspace{6pt}
\noindent where $\phi_0$ is the detailed equivalent of $\phi$, $D_1$, $\ldots$, $D_n$ are the dimension tables of the underlying star schema and the natural joins are performed on the respective surrogate keys. \footnote{This assumes identical names for the surrogate keys; in practice, we use INNER joins along with the appropriate columns of the underlying tables, which might have arbitrary names.}

The expression characterizing a cube query has the following formal semantics\footnote{With the kind help of Spiros Skiadopoulos}:
\[\begin{array}{l}
q = \mbox\{\ x\ |
\mbox(\exists y \in \phi^{0}(DS^0))\ (x=(l_1 = anc^{L_1}_{L^0_1}(y[L^0_1]), \ldots, l_n = anc^{L_n}_{L^0_n}(y[L^0_n]),\ \\ agg_1\{G_1(l_1~\ldots~l_n)\}, \ldots, agg_m\{G_m(l_1~\ldots~l_n)\})\mbox)\mbox\ \}
\end{array}\]
where for every $i$ ($1$ $\leq$ $i$ $\leq$ $m$) the set $G_i$ is defined as follows:
\[\begin{array}{l}
G_i(l_1~\ldots~l_n) = \mbox\{\ m^*\ | \mbox(\exists z \in \phi^{0}(DS^0))\ (l_1=anc^{L_1}_{L^0_1}(z[L^0_1]), \ldots, l_n=anc^{L_n}_{L^0_n}(z[L^0_n]),
       m^*=z[M^0_i]\mbox)\ \mbox\}
\end{array}\]


A cube query specifies (a) the cube over which it is imposed, (b) a selection condition that isolates the
facts that qualify for further processing, (c) the grouping levels, which determine the coarseness of the result, and
(d) an aggregation over some or all measures of the cube that accompanies the grouping levels in the final result.  \\

\result{Interestingly, a cube query carries the typical duality of views: it is, at the same time, both a query, as it involves a query expression imposed over the underlying data, but, also a cube, as it computes a set of cells as a result that obey the constraints we have imposed for cubes.}   \\

Notation-wise, since a query result is also a cube, we use the expression $c$ $\in$ $q$ when a cell belongs to the result of a query, and the expression $q.cells$ to refer to the set of tuples of the result of a query.\\

\result{In the rest of our deliberations, and unless otherwise specified, the  selection conditions of the queries are simple: i.e., they involve a single set-valued equality atom per dimension.}

A note here is due, for the existence of a single atom per dimension in the selection condition of the cube query.
As already mentioned, both $true$ and $D.ALL~=~all$ are both selection conditions, and in fact, in any valid query posed on a specific detailed cube, the result is the same: all the members of the dimension are eligible for the subsequent processing in the query. Despite this, the two expressions are not identical in terms of semantics, and in fact, their automatic translation to a relational query would be different (whereas $true$ implies no atom in the respective SQL query, $D.ALL~=~all$ would induce an extra, unnecessary join); however, a simple cube-to-sql translator would easily take care of the matter. Unless otherwise specified, we assume $D.ALL~\in`~\{all\}$ to be the expression of choice, for reasons of uniformity: this trick allows us to assume simple selection conditions without exceptions, and, with exactly one atom per dimension.

\textbf{Sessions}. A \textit{session} $Q^{S}$ is a list of cube queries $Q^{S}$ = \{$q_{1}$, ..., $q_{n}$\} that have been recorded. We assume the knowledge of the syntactic definition of the queries, and possibly, but not obligatorily, their result cells.

\textbf{History}. A \textit{session history of a user} is a list of sessions. The linear concatenation of these sessions results in a derived session, i.e., a list of queries, following the order of their sessions. The transformation is useful, in order to be able to collectively refer to the \textit{history} of a user as list of queries.


\subsection{Example}\label{sec:example}
Assume a tax office has a cube on the income tax collected and the effort invested to collect it on its allocated citizens. Due to anonymization, the tax office analyst is presented with a detailed cube without the identity of the citizens and has some (pre-aggregated) information along the following dimensions: $Date$, $Work Class$, and $Education$, and two measures $TaxPaid$ by the citizens in thousands of Euros and $HoursSpent$. Each dimension is accompanied by hierarchies of dimension levels. Figure~\ref{fig:BC} depicts the detailed cube data.

\begin{figure}
  \centering
    \includegraphics[width=0.9\textwidth]{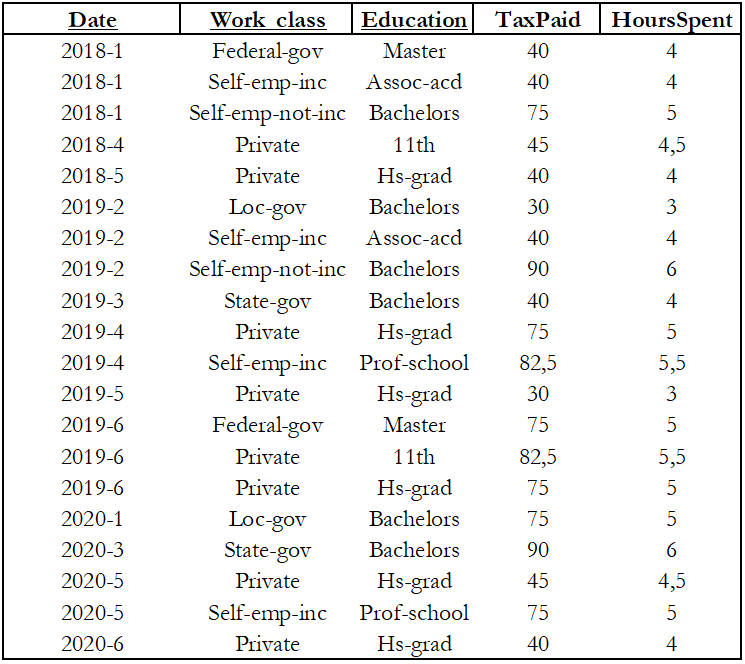}
      \caption{A basic cube}
      \label{fig:BC}
\end{figure}

\textit{Date} is organized in \textit{Months}, \textit{Quarters}, \textit{Years} and \textit{ALL}. \textit{Education} has 5 levels, and \textit{Workclass} 4 levels, and their values, along with their ancestor relationships are depicted in Figure~\ref{fig:dim}. Note that wherever the dimension levels are depicted without values, a surrogate value identical to their ancestor is patched to the dimension, which means that all the dimensions and the value hierarchies are fully defined at all levels at the instance level. 

\begin{figure}
  \centering
    \includegraphics[width=0.9\textwidth]{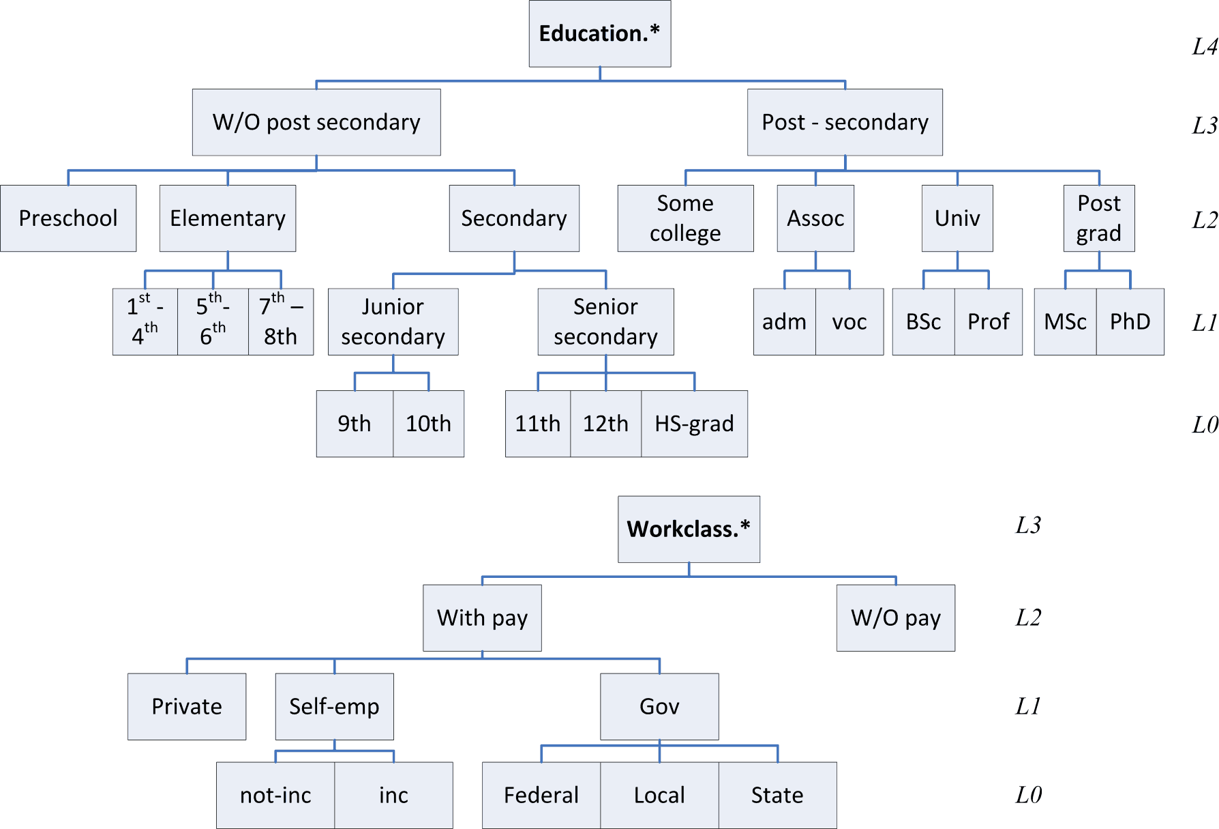}
      \caption{Dimensions}
      \label{fig:dim}
\end{figure}

The idea of ancestor and descendent values is depicted also in the structure of the dimension. So, for example, the \textit{Education} of a group of persons who have attended school till the 11th grade, is characterized with respect to different levels of abstraction as (a) \textit{Detailed}: 11th-grade, (b) \textit{Level 1}: Senior secondary, (c) \textit{Level 2}: Secondary, and (d) \textit{Level 3}: Without Post Secondary. \\

The detailed dataset $\mathbf{DS}^{0}$ defined over these dimensions is, and with a schema $\mathbf{DS}^{0}[\underline{D.L_0, W.L_0, E.L_0}, TaxPaid, HoursSpent]$.\\

A query that can be posed to the aforementioned detailed data set can be:
\[q = 
\tuple{
    \mathbf{DS}^{0},\ \phi,\ [Month, W.L_1, E.ALL, sumTaxPaid],\ [sum(TaxPaid)]\ 
}
\]

with $\phi$ expressed as 
\[\phi = Year \in \{2019,2020\} \wedge W.L_2 \in \{with-pay\}
\]

and actually implying an expression with a single atom per dimension in the form:
\[\phi = Year \in \{2019,2020\} \wedge W.L_2 \in \{with-pay\} \wedge Education.ALL \in \{all\}
\]

\begin{figure}
  \centering
    \includegraphics[width=0.75\textwidth]{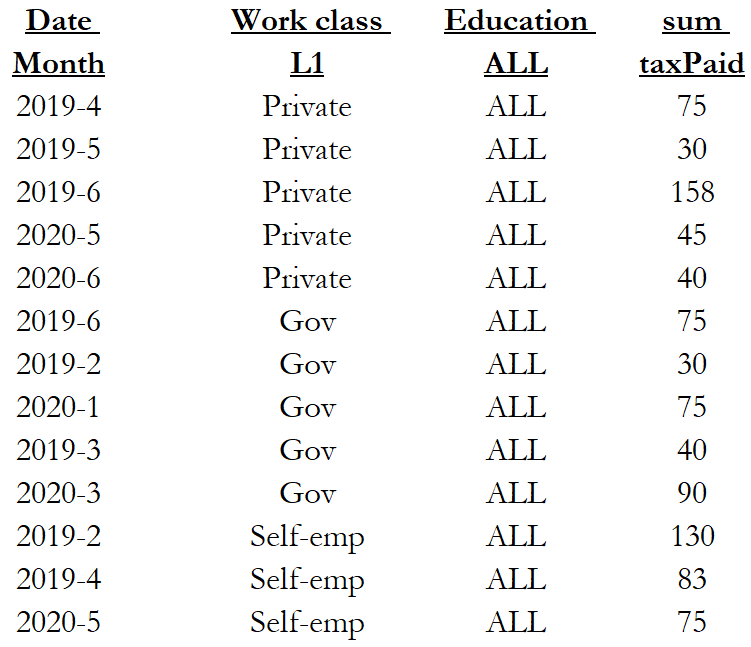}
      \caption{Result of query $q$}
      \label{fig:queryResult}
\end{figure}

\subsection{All the typical OLAP operations are possible}
A key contribution of defining a cube query with the duality of a view, as an expression over a basic cube is that all the typical {OLAP} operations are possible via simple cube queries. The following list gives a set of important examples. In the rest of the deliberations of this subsection, we will assume the existence of the following constructs:

\begin{itemize}
    \item Let $\mathcal{Q}$ be the domain of all (well-formed) query expressions. All operators introduced in this part will be of the form $op: \mathcal{Q} \rightarrow \mathcal{Q}$, i.e., the return a new query expression as the result.
    \item A detailed data set $\mathbf{DS}^{0}$ under the schema $[L_1^0,\ldots,L_n^0, M_1^0,\ldots,M_M^0]$
    \item The most recent query that has been executed, resulting in a query cube, specifically:

\[q = \tuple{\ \mathbf{DS}^{0},\ \phi,\ [D_1.L_1,\ldots, D.L, \ldots, D_n.L_n, M_1,\ldots,M_m],\
[agg_1(M^0_1),\ldots,agg_m(M^0_m)]\ }
\]

\end{itemize}

\textbf{Roll-up}. Assume that for a certain dimension, say $D$, we want to change the level of aggregation to higher level, say $L'$, s.t., $L \preceq L'$. Then, the operator $\mathsf{R\_U}(q, D, L')$ returns the query

\begin{multline*}
\mathsf{R\_U}(q, D, L') = \langle\ \mathbf{DS}^{0},\ \phi,\ [D_1.L_1,\ldots, D.L', \ldots, D_n.L_n, M_1,\ldots,M_m], \\
    [agg_1(M^0_1),\ldots,agg_m(M^0_m)] \rangle 
\end{multline*}


Intuitively, we specified which dimension requires an increase at the level of coarseness, and which level this might be, and the operator  operator $\mathsf{R\_U}(q, D, L')$ returns the respective query expression for obtaining it.\\

\textbf{Drill-down}. This is exactly the symmetric operator of roll-up, where the new level $L'$ is at a lower level than $L$, i.e.,  $L' \preceq L$. Again, the operator $\mathsf{D\_D}(q, D, L')$ produces the query
\[\mathsf{D\_D}(q, D, L') = \tuple{\ \mathbf{DS}^{0},\ \phi,\ [D_1.L_1,\ldots, D.L', \ldots, D_n.L_n, M_1,\ldots,M_m],\
[agg_1(M^0_1),\ldots,agg_m(M^0_m)]\ }
\]
The operator is simply lowering the level of detail for the specified dimension, at the specified level. \result{Observe that the definition of cube queries as expressions over the detailed space is the feature of the model that allows this smooth definition of drill-down} (contrasted to other approaches that avoid retaining the link of a cube to the detailed data that form it).\\

\textbf{Slice (selection)}. Assume we want to apply an extra filter, say $\phi^a$ to the resulting cube $q$. Then, we the operator $\mathsf{Slice}(q, \phi^a)$ returns the query
\[\mathsf{Slice}(q, \phi^a) = \tuple{\ \mathbf{DS}^{0},\ \phi \wedge \phi^a,\ [D_1.L_1,\ldots, D.L, \ldots, D_n.L_n, M_1,\ldots,M_m],\
[agg_1(M^0_1),\ldots,agg_m(M^0_m)]\ }
\]
This allows the introduction of and extra selection condition over the existing cube.\\

\textbf{Projection of Measures}. Assume one wants to change the set of measures of the cube to a new one, retaining some of the previous measures, removing some others, and adding some new ones.  Assume we want to add a new measure $M$ with an aggregate function $agg$ to the cube $q$. This is done via the operator:

\begin{multline*}
\mathsf{AddMeasure}(q, M, agg(M^0)) = \\
\langle\ \mathbf{DS}^{0},\ \phi,\ [D_1.L_1,\ldots, D.L, \ldots, D_n.L_n, M_1,\ldots,M_m,M],\\
[agg_1(M^0_1),\ldots,agg_m(M^0_m), agg(M^0)]\ \rangle 
\end{multline*}

Assume now we want to remove an arbitrary measure (for simplicity, here: $M$) from $q'$. Then, we need to issue $q$ and get the new result.

The operators \textsf{AddMeasure($q$, $M$, $agg(M^0)$)} and 
\textsf{RemoveMeasure($q$, $\{M\}$)} add a new measure and its aggregate function, and remove an old measure from the the schema of q, respectively.\\ 

\textbf{Drill-Across (frequently referred to as cube join)}. Assume now that we have two cubes, defined at the same level of abstraction over the same detailed data set and we want to combine the two cubes in a single result. So, assume the existence of two cubes, which, for simplicity of notation, we will assume with a single measure each (this is directly extensible to multiple measures)

\[q^a = \tuple{\ \mathbf{DS}^{0},\ \phi^a,\ [D_1.L_1,\ldots, D.L, \ldots, D_n.L_n, M_a],\
[agg_a(M^0_a)]\ }
\]
and 

\[q^b = \tuple{\ \mathbf{DS}^{0},\ \phi^b,\ [D_1.L_1,\ldots, D.L, \ldots, D_n.L_n, M_b],\
[agg_b(M^0_b)]\ }
\]

Then, the operator $\mathsf{DrillAcross}_{\bowtie}(q^a, q^b)$ constructs their join, producing a single cube is obtained as
\[\mathsf{DrillAcross}_{\bowtie}(q^a, q^b) = \tuple{\ \mathbf{DS}^{0},\ \phi^a \wedge \phi^b,\ [D_1.L_1,\ldots, D.L, \ldots, D_n.L_n, M_a, M_b],\
[agg_a(M^0_a), agg_b(M^0_b)]\ }
\]

We refer to the above result as the \textsf{Common-Base Inner Join} variant of the drill-across and allows the introduction of the operator $\mathsf{DrillAcross}_{\bowtie}(q^a, q^b)$.  \\

\textit{Drill-Across Variants}. The \textsf{Common-Base Inner Join} variant practically produces the common subset of results of the two cubes, and acts muck like a relational inner join. A most commonly encountered application of this version of drill-across is the case with identical selection conditions for the two cubes.

There are other variants, where the join of the two cubes comes with "outer-join" variants (that require the merging of the two selection conditions on a per-dimension basis) that we do not discuss here. Similarly, the case of different-base drill-across, where the two cubes are defined over different detailed data sets, requires that the two data sets are defined over the exact same multidimensional space (otherwise, we have semantic discrepancies) and then, we need to create a relational view that joins them (i.e, the view has the same dimensions and the union of detailed measures), in order to rebase the result on top of it. Although, this is completely doable by extending the operations applicable at the data set level, this variant is also completely scope of this paper. \\

\textbf{Set operations}. For set operations between two cube queries to be valid, the only difference they can practically have is on their selection condition -- the grouping levels and the aggregate measures have to be the same for the set operations to have any meaning in the first place. This is also very much in line with the relational tradition, where again, set operations are applicable to relations with the same schema.
Let us assume now, we have two cube queries $q^a$ and $q^b$, both under the same schema, along the lines of

\[q = \tuple{\ \mathbf{DS}^{0},\ \phi,\ [D_1.L_1,\ldots, D.L, \ldots, D_n.L_n, M],\
[agg(M^0)]\ }
\]
\noindent and the only difference is that the two queries $q^a$ and $q^b$ come with the selection conditions $\phi^a$ and $\phi^b$, respectively.

For union, the new cube query expression must have a new selection condition $\phi^{new}$ = $\phi^a \lor \phi^b$. Practically, if all atoms are in the form $\alpha: L \in V$, each new atom must be in the form  $\alpha^{new}: L \in V^{a} \cup V^{b}$, if the two atoms of $\phi^a$ and $\phi^b$ are at the same level, or $\alpha^{new}: L^0 \in V^{a^{0}} \cup V^{b^{0}}$, otherwise, with $V^{x^0}$ referring to the set of descendants of the values of $V$  at the lowest possible level of detail.

For intersection, (a) instead of the disjunction of the two selection conditions, we would employ the conjunction, and, (b) at the level of set-valued atoms, instead of the union of the value-sets, we would take the intersection.

For difference, (a) $\phi^{new}$ = $\phi^a \land \lnot \phi^b$, and, (b) the difference of the set-valued atoms  produce the expression for the new cube query.
\newpage

\section{Equivalent expressions for referring to subsets of the multidimensional space}\label{sec:DescSign}
In this Section, we deal with two fundamental characteristics of the multidimensional space: (a) the fact that the same data can be viewed from different levels of detail, and (b) the fact that each query in the multidimensional space applies a border of values of the dimensions, thus "framing" a subset of the space. In this section, we define the terminology and equivalences, that will facilitate the discussion and proofs in subsequent sections.

In a nutshell, an intuitive summary of the ideas and terminology used here can be delineated as follows:
\begin{enumerate}
    \item \result{A proxy is an equivalent expression at a different level of detail that by construction covers exactly the same subset of the multidimensional space, albeit at different level of coarseness}. For example, the detailed proxies of an aggregated cell at the most detailed level are all these cells whose dimension members belong to the most detailed level of the respective dimensions, and which are actually aggregated to produce the aggregate cell of reference. The detailed proxy of a query expression is an expression whose schema is at the most detailed level for each of the dimensions participating in the schema of the query, and whose selection condition is equivalent to the one of the query, but at the most detailed level. Moreover, apart from 'the most detailed level', proxies are definable at arbitrary levels of coarseness. Observe also that proxies are of the same type as their "arguments": the proxy of a cell is a set of cells, the proxy of an expression is an expression, and so on.
    \item \result{The signature of a construct is a set of coordinates that characterize the subset of the multidimensional space "framed" by the construct}.  For example, the signature of a cell are its coordinates at the level of coarseness that the cell is defined, whereas the signature of its detailed proxy are the coordinates produced by the Cartesian product of the descendant values of these coordinates at the most detailed level. Similarly, the signature of a selection condition is the set of coordinates of the multidimensional space for which the selection condition evaluates to true.
    \item \result{Areas are sets of cells within the bounds of a signature}. For example, for a given query $q$, the expression $q.cells$ refers to the cells belonging to the result of the query and $q^0.cells$ is the detailed area of the query, referring to the cells of the most detail level that produce the query result.
\end{enumerate}
In the rest of this section, we will define the above notions rigorously and address algorithmic challenges related to them. Specifically, in section~\ref{sec:detAreas} we define proxies, signatures and related concepts rigorously, and in section~\ref{sec:sign} we present the computation of signatures for various constructs of the model.

\subsection{Transformations: Descendant Proxies, Signatures, Coordinates and Areas}\label{sec:detAreas}
In this section, we define some necessary transformations of expressions, as well as the notation that we will employ, that produce their equivalent expressions at lower levels of the involved dimensions, including the most detailed ones.
To the extent that all computations base their semantics to a query posed at the most detailed levels of a detailed cube $C^0$, providing equivalent transformations is necessary to guarantee correctness.

\subsubsection{Descendant/Detailed proxies of a value}
Assume a value $v$, $v$ $\in$ $dom(L^{h})$. The descendant proxies of $v$ at a level $L^{l}$ are the values of the set $v^{@L^l}$ = $desc_{L^h}^{L^l}(v)$. 

The detailed proxies of value $v$ at level $L^0$, denoted as $v^0$, is  $v^0$ = $v^{@L^0}$ = $desc_{L^h}^{L^0}(v)$.

\subsubsection{Descendant/Detailed proxies of a cell}
Assume a cell $c$ with the values, $c$ = $[v_1, \ldots, v_n, m_1, \ldots, m_m]$ under the schema $[L_1^h, \ldots, L_n^h, M_1, \ldots M_m]$.

The \textit{coordinates}, or \textit{coordinate signature} of a cell $c$, denoted as $c^+$, is the set of level values $[v_1, \ldots, v_n]$. 

The \textit{descendant signature} of the cell $c$ at levels $[L_1^l, \ldots, L_n^l]$, s.t. $L^l_i$ $\preceq$ $L^h_i$ for all $i$, are the values of the Cartesian Product 
$desc_{L^h}^{L^l}(v_1)$~$\times$~$\ldots$~$\times$~ $desc_{L^h}^{L^l}(v_n)$. This set of coordinates is denoted as $\mathbf{c}^{@[L_1^l, \ldots, L_n^l]+}$.

The \textit{descendant proxies} of the cell at levels $[L_1^l, \ldots, L_n^l]$, s.t. $L^l_i$ $\preceq$ $L^h_i$ for all $i$, are denoted as $\mathbf{c}^{@[L_1^l, \ldots, L_n^l]}$, and are the cells in $dom(L_1^l)~\times~\ldots~\times~dom(L_n^l)$ whose coordinates belong to the descendant signature of the cell $v$ at levels $[L_1^l, \ldots, L_n^l]$.

When all $L_i$ are at the most detailed level, we have the detailed proxies of a cell. Equivalently:  the \result{detailed proxies of the cell} $c$, also known as \result{the detailed area of the cell} is the set of detailed cells $\mathbf{c}^0$ = $\{c_1^0, \ldots, c_w^0\}$, where the \textit{coordinates} of each such cell $c_i^0$  are defined as $[\gamma_1, \ldots, \gamma_n]$ over the levels $[L_1^0, \ldots, L_n^0]$, with each $\gamma_j$ $\in$ $\{desc_{L_i^h}^{L^0_i}(v_i)\}$.

\subsubsection{Descendant/Detailed proxies of an atom}

Assume an atom $\alpha$ over a dimension level $D.L$.The \textit{descendant proxy} of atom $\alpha$ at level $L^l$, $L^l$ $\preceq$ $L$, denoted as $\alpha^{L^l}$, is an expression defined as follows, depending on the definition of $\alpha$: 
\begin{itemize}
    \item $\alpha$ is a Boolean atom, or is of the form $D.ALL$ = $all$, or, $D.ALL$ $\in$ $\{all\}$; in this case, the atom remains as is
  \item $\alpha$: $L$ = $v$, $v \in dom(L)$ is transformed to  $\alpha^{@L^{l}}$: $L^l$ $\in$ $V^l$, $V^l$ = $desc^{L^l}_{L}(v)$
    \item $\alpha$: $L$ $\in$ $V$, $V$ =\{$v_1$, \ldots, $v_k$\}, $v_j \in dom(L)$, is transformed to $\alpha^{L^l}$: $L^l$ $\in$ $V^{L^l}$, $V^{L^l}$ = $\bigcup\limits_{j=1}^{k} desc^{L^l}_{L}(v_j)$
\end{itemize}
When level $L^l$ is $L^0$ we refer to the \textit{detailed proxy} of an atom, denoted as $\alpha^0$. 

\subsubsection{Descendant/Detailed proxies of a selection condition}
Assume $\phi$ is a conjunction of \textit{selection~atoms} which are in one of the aforementioned forms, each atom $\alpha_i$ involving a level $L_i^h$. Unless otherwise stated, assume that all $n$ dimensions of a multidimensional space participate, each with a single level, in the expression of $\phi$.

Then, the \textit{descendant proxy} of $\phi$ at levels $[L_1^l, \ldots, L_n^l]$, s.t. $L^l_i$ $\preceq$ $L^h_i$ for all $i$, is denoted as $\phi^{@L_1^l, \ldots, L_n^l}$ and is a selection condition, whose expression is defined as the conjunction of the different $\alpha^{@L_i^l}$.
Practically, assuming that each atom is of the form $L$ $\in$ $V$, $V$ =\{$v_1$, \ldots, $v_k$\}, $v_j \in dom(L)$, the Cartesian Product of all the $n$ $V_i$ sets, produces a set of coordinates for the respective descendant proxy of a selection condition, which we call \textit{descendant signature} of $\phi$.

The \textit{detailed proxy of a selection condition}, $\phi^0$, is an expression produced by placing the most detailed level of each dimension, say $L^0$, in the role of $L^l$. 

The \result{detailed signature of a selection condition} is, therefore, a set of detailed coordinates that construct a boundary at the most detailed level of the cells of the multidimensional space that pertain to the selection condition $\phi$. Therefore, assuming that $\phi$ = $\bigwedge_{i=1}^{n}~\alpha_i$, $\alpha_i:~L_i~\in~V_i$, then, $\phi^0$ = $\bigwedge_{i=1}^{n}~\alpha_i^0$ (with the $\alpha_i^0$ produced as mentioned two subsubsections ago), and the detailed area of $\phi$ is a set of coordinates $\phi^{0^+}$ = $\{\gamma_1, \ldots, \gamma_l\}$, each $\gamma_i$ = $[v_1, \ldots, v_n]$, with each $v_i$, $v_i~\in~dom(L_i^0)$ and $v_i~\in~V^0_i$.

We refer the reader to the Section~\ref{sec:selectionSignature} for an algorithm to compute the signature of a selection condition.

\subsubsection{Descendant/Detailed proxies of a Cartesian Product of coordinates}
Assume a Cartesian Product of sets of coordinates, each set belonging to a different level, say under the expression
$L_i$ $\in$ $V_i$, $V_i$ =\{$v_1$, \ldots, $v_{k_i}$\}, $v_j \in dom(L_i)$.

Assuming a Cartesian Product $X$ defined at levels $[L_1 \ldots L_n]$, the \textit{descendant proxy of} $X$ at levels $[L_1^l, \ldots, L_n^l]$, $L^l_i$ $\preceq$ $L_i$ for all $i$, which we call $X^{@L_1^l, \ldots, L_n^l}$ is produced by substituting each value-set $V_i$ defined at level $L_i$ to a value-set defined at a level $L^l$, as $V^l$ = $\bigcup\limits_{j=1}^{k_j} desc^{L^l}_{L}(v_j)$ and taking their Cartesian Product.

When referring to the most detailed level, the Cartesian Product $X$: $V_1$ $\times$ $\ldots$ $\times$ $V_n$ produces a set of detailed coordinates, $X^0~=~X^{@L_1^0, \ldots, L_n^0}$ that induces an area of coordinates at the most detailed levels of the multidimensional space. 
\begin{remark}
We extend terminology to cover not only coordinates, but also their cells; hence, we say that a Cartesian Product of coordinate values (and thus, a selection condition, too) induces the set of cells whose coordinates are produced by the Cartesian Product. 

\end{remark}

\subsubsection{Descendant/Detailed proxies of a query}

Assume a query $q$ defined as follows:
\[q = \tuple{\ \mathbf{DS}^{0},\ \phi,\ [L_1,\ldots,L_n,M_1,\ldots,M_m],\
[agg_1(M^0_1),\ldots,agg_m(M^0_m)]\ }
\]
Then, \textit{the descendant proxy of the query}, $q^{@L_1^l,\ldots,L_n^l}$ is defined as follows:

\[q^{@L_1^l,\ldots,L_n^l} = \tuple{ 
  \mathbf{DS}^{0},\ \phi^{@L_1^l,\ldots,L_n^l},\ [L_1^l,\ldots,L_n^l,M_1^l,\ldots,M_m^l],\
   [agg_1(M^l_1),\ldots,agg_m(M^l_m)]\ 
}
\]

The \textit{detailed proxy} of the query, $q^0$, is defined for the case where all levels are defined at the lowest possible level $L^0_i$ for all $i$. In this case, since each cell uniquely identifies a single measure value for each $M_i$, the aggregate function is the simple identity function (or equivalently, $min$ or $max$).\\

The \textit{descendant area of the query}, $q^{@L_1^l,\ldots,L_n^l}.cells$, is the set of cells belonging the result of the query $q^{@L_1^l,\ldots,L_n^l}$.

The \result{detailed area of the query}, $q^{0}.cells$, refers to the cells of the result of the query $q^0$.

We refer the reader to the Section~\ref{sec:querySignature} for an algorithm to compute the signature of a query.

\hrulefill
\begin{example}
Assume a query (the red-lettered atoms for Education can be implied) 


\begin{multline*}
q = \langle \mathbf{DS}^{0}, \\
Date.Year \in \{2019,2020\} \wedge Workclass.L2 \in \{With-pay\}\\
\textcolor{red}{\wedge Education.ALL \in \{All\} } \color{black} ,\\ 
[Month,Workclass.L1,Education.ALL,SumTax], [sum(TaxPaid)] \rangle \\
\end{multline*}

Then, the detailed proxy of the query is
\begin{multline*}
q^0 = \langle \mathbf{DS}^{0},\\ 
Date.Month \in \{2019-01,\ldots,2020-12\} \wedge Workclass.L0 \in \{private, not-inc, inc, federal, local, state\},\\
\textcolor{red}{\wedge Education.L0 \in \{Preschool, \dots, PhD\} } \color{black} ,\\ 
[Month,Workclass.L0,Education.L0,TaxPaid], [sum(TaxPaid)] \rangle \\
\end{multline*}    


Observe how the atom $\alpha$: $Date.Year \in \{2019,2020\}$ produces:
\begin{itemize}
    \item a signature $\alpha^+$: $\{2019,2020\}$
    \item a detailed proxy $\alpha^+$: $Date.Month \in \{2019-01,\ldots,2020-12\}$
    \item a detailed signature $\alpha^{0^+}$: $\{2019-01,\ldots,2020-12\}$
\end{itemize}

\end{example}
\hrulefill


\subsubsection{Summary of notation and concepts}

\begin{table}[h!]
\centering

\begin{tabular}{|p{2cm}|p{2cm}|p{2.5cm}|p{2.5cm}|p{2.5cm}|} 

 \hline
 
& \multicolumn{2}{l|}{\textit{Signature: tuple of dimension values}} & \multicolumn{2}{l|}{\textit{Proxy(x): of the same type as x}} \\
\hline
    &\textit{Coord.~signature} or \textit{coordinates} & \textit{Detailed Signature} & \textit{Descendant Proxy} & \textit{Detailed Proxy}\\ 
\hline
value $v$          
    &   
    & 
    & $v^{@^\mathbf{L}}$ = set of values at desc.~level
    & $v^{0}$: set of values at zero level
    \\ 
\hline
set of coordinates $X$ 
    &                                             
    &                                                                                 
    & $X^{@\mathbf{L}}$: set of coord.~at \textbf{L}             
    & $X^0$: set of coordinates at zero level
    \\ 
\hline
atom $\alpha$               
    &$\alpha^+$: set of dim.~values qualifying the atom, at the level of $\alpha$
    &$\alpha^{0^{+}}$: set of dim.~values qualifying the atom, at the zero level 
    &$\alpha^{@\mathbf{L}}$: equiv.~expression at desc.~levels
    &$\alpha^0$: equiv.~expression at zero level   
    \\ 
\hline
condition $\phi$     
    &$\phi^{+}$: coordinates produced by the Cart.~Prod.~of the $\alpha^+_i$
    &$\phi^{0^{+}}$: coord.~ produced by the Cart.~Prod.~of the $\alpha^{0^{+}}_i$
    &$\phi^{@\mathbf{L}}$: equiv.~expression at desc.~levels                    
    &$\phi^{0}$ equiv.~expression at zero level                           
    \\ 
\hline
cell $c$
    &$c^+$: tuple of cell's dim.~values  
    &$c^{0^{+}}$: set of coord.~of detailed proxy 
    &$c^{@\mathbf{L}}$: desc.~area = set of cells at lower level 
    & $c^0$: detailed area = set of cells at zero level             
    \\ 
\hline
query expression $q$              
    &$q^{+}$: set of coordinates of cube cells
    &$q^{0^{+}}$: set of coord.~of detailed proxy            
    &$q^{@\mathbf{L}}$: equiv.~query expression at \textbf{L}                            
    &$q^0$: equiv.~expression at zero level                                
    \\
\hline
\end{tabular}
\caption{Notation and central notions. Proxies (both descendant and detailed) are of the same type as their subject. Coordinate signatures are (sets of) \emph{tuples of dimension level values}. Areas of cells and queries are sets of \emph{cells}(e.g., for a query $q$, $q.cells$ is the area of the query and $q^0.cells$ is the detailed area of the query )
}
\label{tab:notation}
\end{table}

In Table \ref{tab:notation} we provide a summary of notation as well as a short reminder of the type of each of the important concepts involved so far in our discourse.


\clearpage

\subsection{Working with signatures}\label{sec:sign}

\subsubsection{Computing the signature of a selection condition}\label{sec:selectionSignature}
Assume we have a selection condition $\phi$ and we want to compute its signature $\phi^+$. How can we do that?

\begin{algorithm}[H]
\DontPrintSemicolon 
\KwIn{An selection condition $\phi$ having a single atom per dimension}
\KwOut{The signature of the selection condition $\phi^{+}$ and its detailed proxy $\phi^{0^+}$}
\Begin{
    \ForAll {dimensions $D$ without an atom in $\phi$}{
        introduce an atom $D.ALL$ $\in$ $\{D.all\}$\;
    }
    Convert all atoms of $\phi$ in the form $\alpha$: $D.L$ $\in$ $V$, $V$ = $\{v_1, \ldots, v_k\}$ \;
    $\phi^{+}$ $\gets$ $V_1$ $\times$ $V_2$ $\times$ $\ldots$ $\times$ $V_n$ \;
    Convert all $V$ to their detailed proxies $V^0$, with each $v_i~\in~V$ being replaced by its detailed proxy $v_i^0$, $v_i^0=desc_L^{L^0}(v_i)$\; 
    $\phi^{0^+}$ $\gets$ $V_1^0$ $\times$ $V_2^0$ $\times$ $\ldots$ $\times$ $V_n^0$ \;
    \Return{$\phi^{+}, \phi^{0^+}$}
}
\caption{\sf{Compute the Signature of a Selection Condition}}
\label{algo:ComputeSignatureOfSelectionCondition}
\end{algorithm}

\subsubsection{Grouper domains of atoms and selection conditions}\label{sec:grouperDomain}
Assume that we have an atom which is going to be used as a filter of a query, to be posed upon a detailed data set, in order to restrict the range of participation to the query result. Let's assume that the expression of the atom is defined at a certain level $D.L^{\phi}$. Being a part of a query, the detailed cells that fulfil the atom's criterion will then be grouped by a level $D.L^{g}$, which is probably different that the selection level. We do this for every dimension, and we can compute the signature of the query. The question is then: what are exactly the values of each dimension that will appear in the query result?

\textbf{Grouper domain of an atom}. Assume we have an atom of the form $\alpha$: $D.L^{\phi}$ $\in$ $V$, $V$ = $\{v_1, \ldots, v_k\}$ and we want to compute what will be the resulting set of values if a grouper $D.L^{g}$ is applied to them. We define the \emph{grouper domain} of an atom $\alpha$: $D.L^{\phi}$ $\in$ $V$, $V$ = $\{v_1, \ldots, v_k\}$ with respect to a grouper level $D.L^{g}$ of the same dimension as follows:

\begin{equation}
  gdom(\alpha, D.L^{g})=\begin{cases}
    \{desc_{L^{\phi}}^{L^{g}}(v_1), \dots,desc_{L^{\phi}}^{L^{g}}(v_k) \}, & \text{if $L^{g}$ $\preceq$ $L^{\phi}$}\\
    \{anc_{L^{\phi}}^{L^{g}}(v_1), \dots,anc_{L^{\phi}}^{L^{g}}(v_k) \},  & \text{otherwise}.
  \end{cases}
\end{equation}
  
Equivalently, we can also express an atom's grouper domain as:
\begin{equation}
  gdom(\alpha, D.L^{g})=
    \{anc_{L^{0}}^{L^{g}}(desc_{L^{\phi}}^{L^{0}}(v_1)), \dots, anc_{L^{0}}^{L^{g}}(desc_{L^{\phi}}^{L^{0}}(v_k)) \}
\end{equation}

For example, assume $\alpha: Date.Year \in \{2019,2020\}$ and $D.L^{\phi}: Date.Month$ being the grouper level. Then, $gdom(\alpha, D.L^{g})$ = $\{2019-01,\ldots,2020-12\}$

\textbf{Grouper domain of a selection condition}. Assume we have a selection condition expressed as a conjunction of exactly one atom per dimension, for all dimensions involved in a query. Then, \emph{the grouping domain of the selection condition is the Cartesian product of the grouping domains of the individual atoms} and, remarkably, \emph{it is also equivalent to the query signature}.

Assume a query $q$ defined as follows:

\[q = \tuple{ \mathbf{DS}^{0},\ \phi,\ [L_1,\ldots,L_n,M_1,\ldots,M_m],\
[agg_1(M^0_1),\ldots,agg_m(M^0_m)]\ }
\]
with 
$\phi$ = $\alpha_1$ $\wedge$ $\ldots$ $\wedge$ $\alpha_n$.\\

Then,\\

$gdom(\phi, [L_1,\ldots,L_n])$ = $q^{+}$ = $gdom(\alpha_1, D_1.L_1)$ $\times$ $\dots$ $\times$ $gdom(\alpha_n, D_n.L_n)$


\subsubsection{Computing the signature of a query}\label{sec:querySignature}

\silence{
Before producing coordinates for an arbitrary cube, let's start by producing coordinates for a query producing a detailed cube.

\textbf{Detailed level}. Assume a set of dimension levels $S^0$ = $[D_1.L_1^0,\ldots,D_n.L_n^0]$ at the most detailed level. For each dimension, assume a single atom $\alpha_i^0$:  $L_i^0$ $\in$ $V_i^0$, $V_i^0$ =\{$v_1^0$, \ldots, $v_k^0$\}. Then, the Cartesian Product $X^0$: $V_1^0$ $\times$ $\ldots$ $\times$ $V_k^0$ produces a set of detailed coordinates for the selection condition $\phi^0$ which is the conjunction of the individual atoms $\alpha_i^0$.

Assuming a specific dataset $DS^0$, the Cartesian Product $X^0$ \result{produces} a subset of the dataset, as a set of detailed cells, $\Gamma^0$, via the query
\[q^{0\star} = \mathbf{DS}^{0},\ \phi^0,\ [L_1^0,\ldots,L_n^0,M_1,\ldots,M_m],\
[min(M^0_1),\ldots,min(M^0_m)]\ \]

s.t. $\phi^0$ = $\bigwedge_{i=1}^{n}$ $a_i$

\noindent (note that min, here, is used equivalently to the identity function, since the combination of all the $L_i^0$ induces a single detailed cell).\\

\newpage

\textbf{Arbitrary level}. Now, we are ready to define coordinates for a query defined at an arbitrary level of detail. 
}

Assume a query $q$ defined as follows:

\[q = \tuple{ \mathbf{DS}^{0},\ \phi,\ [L_1,\ldots,L_n,M_1,\ldots,M_m],\
[agg_1(M^0_1),\ldots,agg_m(M^0_m)]\ }
\]

with $S$ = $[D_1,L_1,\ldots,D_n.L_n]$ at arbitrary levels of coarseness and $\phi$ a simple selection condition (therefore, for each dimension, assume a single atom  $\alpha_i$:  $L_i$ $\in$ $V_i$, $V_i$ =\{$v_1$, \ldots, $v_k$\}).

\silence{
We can compute the detailed proxy of $\alpha$: $\alpha_i^0$:  $L_i^0$ $\in$ $V_i^0$, $V_i^0$ =\{$v_1^0$, \ldots, $v_k^0$\}, defined at the detailed level (equiv., being the detailed proxy of an atom at an arbitrary level, expressed as a finite set of values). Then, for each level $L$ of $S$ we can pre-determine the set of grouper coordinates produced by $\phi^0$, named $V^{@L}$ as $V^{@L}$ = $\bigcup_{i=1}^{k}$ $anc_{L^0}^{L}(v_i^0)$ (i.e., the members of $V^{@L}$ belong to $dom(L)$).

The Cartesian Product $X$: $V^{@L_1}$ $\times$ $\ldots$ $\times$ $V^{@L_n}$  \result{produces} a set of coordinates $\Gamma^S_{\phi^0}$:$\{ \gamma_1, \ldots, \gamma_\lambda \}$, with each $\gamma_i$ being a tuple of the form $[v_1^i, \ldots, v_n^i]$, with each $v_j^i~\in~dom(L_j)$ (i.e., belonging to the domain of the respective level of the schema $S$ -- attn., not at the detailed level).

The resulting cells will have \emph{exactly} the coordinates of  $\Gamma^S_{\phi^0}$ -- i.e., for each tuple $\gamma~in~X$ there will exist exactly one cell in $q.cells$ with the same coordinates, and vice-versa.

Observe, that as already mentioned, the above setup works for both detailed selection conditions, and selection conditions at arbitrary levels that are translatable to their detailed proxies.\\
}

To produce $q^+$, the coordinates of a query, we can first compute its detailed signature $q^{0^+}$ and then roll-them up to the grouper levels of $q$ -- i.e., we can proceed as follows:
\begin{enumerate}
    \item produce $\phi^0$ from $\phi$;
    \item produce $\phi^{0^+}$ from $\phi^0$ (i.e., the coordinates of $\phi^0$); this is also the detailed area of the query, $q^{0^+}$;
    \item produce $q^+$ as follows: for each detailed $\gamma_i^0$ in $\phi^{0^+}$, for each value $v_j~\in~\gamma_i^0$, replace it with $anc_{L_j^0}^{L_j}(v_j)$ and add the resulting $\gamma_i^{@L_1, \ldots, Ln}$ to the \textit{set} of coordinates $q^+$  
\end{enumerate}

\begin{algorithm}[H]
\DontPrintSemicolon 
\KwIn{A query $q$ having a simple selection condition $\phi$}
\KwOut{The signature of the query $q^{+}$ and its detailed proxy $q^{0^+}$}
\Begin{
    produce $\phi^0$ from $\phi$ \;
    produce $\phi^{0^+}$ from $\phi^0$\;
    $q^{0^+}$ $\gets$ $\phi^{0^+}$ \;
    $q^+$ = $\emptyset$ \;
    \ForAll {$\gamma_i^0$ in $\phi^{0^+}$}{
        \ForAll {value $v_j~\in~\gamma_i^0$}{
            $\gamma_i^{@L_1, \ldots, Ln}$ $\gets$ $anc_{L_j^0}^{L_j}(v_j)$ \;
            $q^{+}.add(\gamma_i^{@L_1, \ldots, Ln})$ \; 
        }
    }

    \Return{$q^{+}, q^{0^+}$}
}
\caption{\sf{Compute Query Signature and Detailed Query Signature}}
\label{algo:ComputeSimplySignatureOfQuery}
\end{algorithm}

Equivalently, Algorithm \ref{algo:SimpleQuerySignature} pursues a different but equivalent transformation, that computes the grouper values per dimension first, and then takes their Cartesian Product. Practically, for each dimension, we compute its grouper domain, and then, we take the Cartesian Product of all grouper domains, resulting in the grouper domain of the selection condition, which is also the signature of the query.

\begin{algorithm}[H]
\DontPrintSemicolon 
\KwIn{A query $q$ at an arbitrary level of detail with a simple selection condition}
\KwOut{The query signature $q^+$}
\Begin{
    \ForAll{ dimensions $D_i$ with atom $\alpha_i:~D_i.L_i^\phi~\in~V$ and grouper $L_i$} 
        {produce the detailed proxy of $\alpha_i$:  $\alpha_i^0:~L_i^0~\in~V^0$ \;
        Let $V^L_i$ be the set of grouper values of $D_i$, $V^L_i$ = $\emptyset$ \;
        \ForAll{$v_j~\in~V_i^0$}{
            $V^L_i$ = $V^L_i$ $\cup$ $anc_{L_i^0}^{L_i}(v_j)$ \;
        }
    }
    $q^{+}$ $\gets$ $V^L_1$ $\times$ $\dots$ $\times$ $V^L_n$ \;
\Return{$q^{+}$ }
}
\caption{\sf {Produce Query Signature}}
\label{algo:SimpleQuerySignature}
\end{algorithm}


\hrulefill
\begin{remark}
Speedups for the above are: (a) if a certain $L_i$ is $ALL$, immediately add $all$ at the respective values; (b) if the selection condition's atom of a dimension is at a lower level than the schema level, there is no reason to first drill down to $L^0$ and then roll-up the values to $L$, but can immediately roll-up the values via $anc_{L_i^\phi}^{L_i}(v)$; (c) on the other hand, if the grouper $L_i$ is lower than the filter $L_i^\phi$, then, we can immediately drill-down the values of $V$ to their $desc_{L_i^\phi}^{L_i}(\cdot)$.

Alternative evaluation plans could include taking $dom(L_i)$ and start disqualifying values that are filtered out due to $\alpha_i$; then taking the the Cartesian Product of the resulting $n$ sets that are now subsets of $dom(L_i)$.
\end{remark}
\hrulefill

\begin{example}
Assume a query 

\begin{multline*}
\noindent q = \langle \mathbf{DS}^{0}, Date.Year \in \{2019,2020\} \wedge Workclass.L2 \in \{With-pay\}, \\
{[Month,Workclass.L1,Education.ALL,SumTax]}, [sum(TaxPaid)] \rangle\\
\end{multline*}

Here, since the atom on Education was not originally specified, it is implied that a 'All' atom applies for Education. We will use it in the sequel to produce signatures. Thus $\phi$ becomes:

\begin{multline*}
\phi: Date.Year \in \{2019,2020\} \wedge Workclass.L2 \in \{With-pay\} \\
\wedge Education.ALL \in \{All\}  \color{black}   
\end{multline*}

Then, the signature, $\phi^+$, of the selection condition $\phi$ is
\[
\phi^{+}: \{2019,2020\} \times \{With-pay\} \times \{All\} = \{\tuple{2019,With-Pay,All},\tuple{2020,With-Pay,All}\}
\]

The detailed selection condition $\phi^0$ is:
\begin{multline*}
\phi^0: Date.Month \in \{2019-01,\ldots,2020-12\} \\
\wedge Workclass.L0 \in \{private, not-inc, inc, federal, local, state\}\\    
\wedge Education.L0 \in \{Preschool, \dots, PhD\} 
\end{multline*}

Then, the respective detailed signature $\phi^{0^+}$ as well as the detailed query signature $q^{0^+}$ is:
\begin{multline*}
\phi^{0^+} = q^{0^+}: \{2019-01,\ldots,2020-12\} \times \{private, not-inc, inc, federal, local, state\} \\
\times \{Preschool, \dots, PhD\} \\
= \{\tuple{2019-01,private,preschool}, \ldots, \tuple{2020-12,state,PhD}\}
\end{multline*}

Coming to the query now, the signature of the query is produced by rolling up the signature of $\phi^0$ to the grouper levels:
\begin{multline*}
q^+: \{2019-01,\ldots,2020-12\} \times \{Private, Self-emp, Gov\} \times \{ALL\}\\    
= \{\tuple{2019-01,Private,ALL}, \ldots, \tuple{2020-12,Gov,ALL}\}
\end{multline*}

Observe that the query signature is expressed as the Cartesian Product of the grouper domains of the individual atoms of the selection condition, i.e., $\{2019-01,\ldots,2020-12\}$ for $Date$,  $\{Private, Self-emp, Gov\}$ for $WorkClass$ and $\{ALL\}$ for $Education$.

Observe also that at the end of the day, all signatures, produced as Cartesian Products of values, are \textit{sets} of \textit{coordinates} (with coordinates being tuples of values with a single value per dimension).
\end{example}
\hrulefill

\subsubsection{Other signature operations}\label{sec:otherSign}
\textbf{Computing the difference/intersection of two signatures}. Given two signatures defined over the same dimensions, both signatures come as sets of coordinates. Then, the well-known set difference computes the difference of the two signatures. Equivalently, set intersection works for the intersection of two signatures.

A simple generic algorithm can take as input (a) a query $q$ being under test, and (b) a benchmark query $q^\star$ against which $q$ is going to be tested and label the signature of $q$ with two characterizations: (i) \emph{covered} coordinates, i.e., coordinates already being part of the signature of $q^\star$, and (ii) \emph{novel} coordinates, i.e., coordinates which are not part of the signature of $q^\star$. The respective sets $q^{cov+}$ and $q^{nov+}$ collect the respective coordinates, and their union produces $q^{+}$.

\begin{algorithm}[H]
\DontPrintSemicolon 
\KwIn{A query $q$ and a benchmark query $q^\star$ defined over the same levels }
\KwOut{The subset of the coordinates of $q$, say $q^{cov+}$ that are already part of the result of $q^{\star}$, and its complement $q^{nov+}$}
\Begin{
produce $q^{+}$ and $q^{\star^{+}}$ \;
$q^{cov+}$ $\gets$ $q^{+}$~$\bigcap$~$q^{\star^{+}}$  \;
$q^{nov+}$ $\gets$ $q^{+}$ - $q^{\star^{+}}$ \;

\Return{$q^{cov+}$, $q^{nov+}$ }\;
}
\caption{\sf{Produce Covered And Novel Query Coordinates}}
\label{algo:QueryIntersectionSignatureDiff}
\end{algorithm}

\hrulefill
\begin{example}
Assume the signature, $\phi^{+}_1$ 
\[
\phi^{+}_1: \{2019,2020\} \times \{With-pay\} \times \{All\} = \{\tuple{2019,With-Pay,All},\tuple{2020,With-Pay,All}\}
\]

and the signature, $\phi^{+}_2$ defined as\\

$\phi^{+}_2: \{2018,2019\} \times \{With-pay, Without-pay\} \times \{All\}$ = 

\begin{multline*}
\{\tuple{2018,With-Pay,All},\tuple{2019,With-Pay,All},\\
\tuple{2018,Without-Pay,All},\tuple{2019,Without-Pay,All}\}
\end{multline*}

The intersection of the two signatures signifies the common part of the multidimensional space they cover: $\phi^{+}_1 \cap \phi^{+}_2: \tuple{2019,With-Pay,All}$. 

The union $\phi^{+}_1 \cup \phi^{+}_2$ of the two signatures signifies the joint subspace the expression
$\phi_1 \lor \phi_2$ covers

\begin{multline*}
\{\tuple{2018,With-Pay,All},\tuple{2019,With-Pay,All},\\
\tuple{2018,Without-Pay,All},\tuple{2019,Without-Pay,All}\\
\tuple{2020,With-Pay,All}
\}
\end{multline*}
Again, observe that signatures are \emph{sets}, specifically, sets of coordinates, and therefore they are treated via set operations.
\end{example}

\hrulefill
\newpage

\newpage
\section{Foundational Containment}\label{sec:foundCont}

\subsection{Preliminaries and Assumptions}
Before proceeding, let us remind the reader of simple selection conditions. Simple selection condition are characterized by the following properties:
\begin{itemize}
    \item a simple conjunction of atoms, $\phi$ = $\bigwedge\limits_{j=1}^{p} a_i$,
    \item all atoms in the selection condition of all the queries are of the form: $D.L$ $\in$ $\{ v_1, \ldots, v_k \}$, $v_i$ $\in$ $dom(L)$
    \item there is exactly one atom per dimension; for the dimensions where no selection atom is defined (equivalently: $true$ is the selection atom), for reasons of the homogeneity we assume the expression $D.ALL~\in~\{all\}$, which effectively incorporates the entire active domain of the dimension.
\end{itemize}

In the rest of all our deliberations, we will assume a query $q^n$ (n for "new" and "narrow") with a simple selection condition $\phi^n$, and a query $q^b$ (b for "broad") with a simple selection condition $\phi^b$. 

\result{The \underline{decision} problem at hand is: given the query $q$ and the query $q^n$, and without using the extent of the cells of the two queries, can we compute whether the cells of the detailed proxy of $q^n$, i.e., the result of $q^{n^{0}}$ is a subset of the result of $q^{0}$, i.e., the detailed proxy of $q$?}

\result{In a similar vein, the respective \underline{inverse enumeration} problem is: can we compute which cells of $q^{n^{0}}$ are not part of $q^{0}$, and which are not?}

\begin{remark}
The aforementioned setup for atoms covers a very large spectrum of commonly encountered cases, like: (a) the case of a point query $L$ = $v$, (b) the case the disjunction of values, expressed via set membership, and, (c) since we assume that dimensions come with finite countable domains (and in fact totally ordered) this setup also covers the case of range-selections, where the atom is of the form $L$ $\in$ $[v_{low} \ldots v_{high}]$.
\end{remark}
\begin{remark}
Observe that the problem is independent of the aggregations and the roll-ups taking place in the queries, and, fundamentally boils down to selection condition comparison.
\end{remark}

\subsection{Foundational Containment}
\begin{definition}
A query $q^b$ \textit{foundationally~contains} a query $q^n$, denoted as $q^n~\sqsubseteq^0~q^b$ if the detailed area of $q^b$ is a superset (i.e., of detailed cells) over the detailed area of $q^n$.

Equivalently: $\forall$ cell $c_i^{0^n}$ in the detailed area of $q^n$, $c_i^{0^n}$ also belongs to the detailed area of $q^b$, too. 
\end{definition}

\begin{remark}
Note that this does not guarantee computability of $q^n$ from $q^b$, due to the intricacies of aggregation; however, it is a necessary condition for assessing computability, as, if the condition fails, there exist detailed cells that pertain to the new query $q^n$ that have not been taken into consideration for the computation of the (potentially pre-existing) $q^b$, and thus computing the former from the cells of the latter is impossible.
\end{remark}

Now, we are ready to give a necessary and sufficient condition for foundational containment to hold.\sideNote{Is my detailed area contained in yours?}
\begin{theorem}
Assume two queries, $q^n$ and $q^b$, having exactly the same dimension levels in their schema and a 1:1 mapping between their measures (obtained via the identity of the respective $agg_i(M_i^0)$ expressions). To simplify notation, we will assume the two queries have the same measure names, and thus, exactly the same schema [$L_1, \ldots L_n, M_1, \ldots, M_m$]. Assume also their respective simple, detailed selection conditions $\phi^{0^n}$ and $\phi^{0^b}$.
Let $\phi^{0^n}$ have atoms of the form $D.L^0$ $\in$ $V^0$, $V^0$ = $\{v_1, \ldots, v_k\}$ and $\phi^{0^b}$ have atoms of the form $D.L$ $\in$ $U^0$, $U^0$ = $\{u_1, \ldots, u_m\}$, for every dimension $D$ pertaining to the two cubes $q^{n}$ and $q^{b}$, respectively. Then, 
$q^{b}$ foundationally contains $q^{n}$ if and only if the following holds:

$\forall$ atom of $\phi^{0^n}$, say $D.L^0$ $\in$ $V^0$: 
    $\forall$ $v_i$ $\in$ $V^0$, $v_i$ $\in$ $U^0$, i.e., $V^0~\subseteq~U^0$

\end{theorem}
\begin{proof}
Assume the above property holds. Then, the cells that belong to the detailed area of $q^n$, produced by the conjunction of $n$ atoms of the form $D_i.L_i^0$ $\in$ $V_i^0$, are produced by the signature obtained by taking the Cartesian product of the values belonging to the value-sets $V_1^0$ $\times$ $V_2^0$ $\times$ $\ldots$ $\times$ $V_n^0$.
The respective detailed signature for $q^b$ is $U_1^0$ $\times$ $U_2^0$ $\times$ $\ldots$ $\times$ $U_n^0$.
If for every pair of value-sets for the same dimension, say $D_i$, $V_i^0~\subseteq~U_i^0$, the Cartesian product produced for $q^n$ is a subset of the Cartesian product produced for $q^b$, i.e.,  $V_1^0$ $\times$ $V_2^0$ $\times$ $\ldots$ $\times$ $V_n^0$ $\subseteq$ $U_1^0$ $\times$ $U_2^0$ $\times$ $\ldots$ $\times$ $U_n^0$  .  
Then, by definition, $q^n$ $\sqsubseteq^0$ $q^b$. 

Inversely, via reductio ad absurdum, assume that $\exists$ $v_j~\in~V_i$, s.t., there does not exist any $u_{j'}~\in~U_i,~u_{j'} = v_j$. Then, all the cell coordinates generated by the participation of $v_j$ in the Cartesian Product will not belong to the $U_1^0$ $\times$ $U_2^0$ $\times$ $\ldots$ $\times$ $U_n^0$ either. Therefore, there will be cells in the detailed area of $q^n$ that do not belong to the detailed area of $q^b$. Absurd.
\end{proof}

\hrulefill
\begin{remark}
Observe that the above is both an adequate and a necessary condition for foundational containment. Thus, producing the detailed selection condition and from this, the detailed signatures of two queries, we can check for foundational containment. To the extent that we have a single atom per dimension, the complexity of the check implied by the above Theorem is linear to the number of dimensions.
\end{remark}
\hrulefill\newline

\subsection{Foundational containment when expressions are complex}\label{sec:fc-complex}
Assume now that instead of dealing with the detailed selection conditions at the most detailed level for all dimensions, we work with selection conditions defined at arbitrary levels. It is true that we can always transform selection conditions at arbitrary levels to their detailed proxies and perform a precise check for foundational containment. But can we do faster? We introduce a sufficient but not necessary condition to perform a fast check. \sideNote{Is my detailed area contained in yours? (fast)}

\begin{theorem}
 
Assume two queries , $q^n$ and $q^b$, having exactly the same dimension levels in their schema and a 1:1 mapping between their measures (obtained via the identity of the respective $agg_i(M_i^0)$ expressions). To simplify notation we will assume the two queries have the same measure names, and thus, exactly the same schema [$L_1, \ldots L_n, M_1, \ldots, M_m$]. Assume also their respective simple selection conditions $\phi^n$ and $\phi^b$, such that $\phi^{n}$ has atoms of the form $D.L^{n}$ $\in$ $V$, $V$ = $\{v_1, \ldots, v_k\}$ and $\phi^{b}$ has atoms of the form $D.L^{b}$ $\in$ $U$, $U$ = $\{u_1, \ldots, u_m\}$, for every dimension $D$ pertaining to the two cubes' schema ($L$ being an arbitrary level of the dimension, and not obligatorily the most detailed one).

Then, $q^b$ foundationally contains $q^n$, $q^n$ $\sqsubseteq^0$ $q^b$, if the following holds: \\
\indent $\forall$ atom of $q^n$, say for the dimension $D$, $D.L^{n}$ $\in$ $V$, $V$ = $\{v_1, \ldots, v_k\}$\\
\indent \indent $\forall$ $v$ $\in$ $V$, $\exists$ $u$ $\in$ $U$ in the respective atom of $q^b$ for $D$, s.t., $u$ = $anc_{L^n}^{L^b}(v)$

\end{theorem}

\begin{proof}
Assume the theorem's condition holds and for each $v$ there exists a correspondence to $u$ = $anc_{L^n}^{L^b}(v)$ in $U$. Then, the detailed proxy of $u$ is a superset of the detailed proxy of $v$. The union of the detailed proxies of the $v_j$ values, is therefore, a subset of the union of the detailed proxies of the respective $u_j$ (even if multiple $v$ values are mapped to the same $u$). Therefore, $V^0$ $\subseteq$ $U^0$.

The above hold even if $L^{n}$  and $L^{b}$ are the same level, and thus, we simply want every value of $V$ to be also present in $U$. This involves the level $ALL$ too. Also, the above holds even if multiple $v$ values are mapped to the same $u$, as due to the monotonicity of domains, even if \emph{all} the descendants of $u$ are present in $V$, the union of their detailed proxies is still a subset of the detailed proxy of $u$ (with equality holding, obviously, in the case of all descendants being present).
\end{proof}

\hrulefill
\begin{remark}
Obviously from the requirement of the theorem, every level $D.L^n$ of $\phi^n$ is lower or equal than the respective level $D.L^b$ of $\phi^b$. This is not necessarily reflected in the schemata of the two cubes, as the selection conditions can take place at arbitrary levels, different from the grouper levels that appear in the schema of the query. But, when selection conditions are concerned, all the levels involved in the narrow query are lower or equal than the respective levels in the broader query.

Note also that due to the fact that the order of levels is a \emph{partial order}, the respective levels of the two selection conditions can be the same.

Also, for every valid value of $q^n$, there must exist a value of $q^b$ that covers a broader span of values.

The inverse of the Theorem does not hold. Assume the case where $\phi^n$: $Continent$ = $Oceania$ and $\phi^b$: $Country$ $\in$ $\{ Australia, New~Zealand, ... \}$ (a superset of the countries of Oceania). Then, although the detailed proxy of Oceania is a subset of the union of the detailed proxies of the countries in the set $U$ of $q^b$, and $V^0$ $\subseteq$ $U^0$ holds, the condition of the Theorem is not met.
\end{remark}
\hrulefill\newline

\begin{lemma}
For the case where both queries have dicing selection conditions, i.e., single-member set-values for each atom of their selection condition, we can say that $q^b$ foundationally contains $q^n$, $q^n$ $\sqsubseteq^0$ $q^b$, if the following holds: \\
\indent $\forall$ atom of $q^n$, say $a$: $D.L^{n}$ = $v$, the respective atom of $q^b$, say $a'$: $D.L^{b}$ = $u$, involves a value $u$ s.t., $u$ = $anc_{L_n}^{L_b}(v)$
\end{lemma}
\begin{proof}
Obvious.
\end{proof}

\hrulefill
\begin{example}
Assume the following two queries with the same schema and different selection conditions.

\[q^o = \tuple{ \mathbf{DS}^{0},\ \phi^o,\ [Month, W.L_1, E.ALL, sumTaxPaid],\
[sum(TaxPaid)]\ } 
\]

having
\[\phi^o = Year \in \{2019,2020\} \wedge W.L_2 \in \{with-pay\}
\]

and
\[q^n = \tuple{ \mathbf{DS}^{0},\ \phi^n,\ [Month, W.L_1, E.ALL, sumTaxPaid],\
[sum(TaxPaid)]\ 
}
\]

having
\[\phi^n = Year \in \{2019\} \wedge W.L_1 \in \{private, self-emp\}
\]

Then, we can see that all the conditions of the theorem are held:
\begin{itemize}
    \item both queries have the same schema;
    \item all the atoms of the two selection conditions are in the form requested by the query (both queries imply an atom of the form $E.ALL \in \{All\}$ too);
    \item for every value appearing in the atoms of $q^n$, there is an ancestor in the value-set of $q^o$ -- specifically, for Year, $anc_{Year}^{Year}(2019)$ = 2019 which is part of the value-set for the atom of $\phi^o$, and, both values of $\{private, self-emp\}$ have an ancestor in $L_2$ which is $with-pay$ (also in the value set of the respective atom in $\phi^o$).
\end{itemize}
Observe also how all the levels of the atoms of $\phi^o$ are at higher or equal height than the ones of $\phi^n$. 
\end{example}
\hrulefill

\begin{figure}
  \centering
    \includegraphics[width=0.5\textwidth]{figures/queryResult}
    \includegraphics[width=0.45\textwidth]{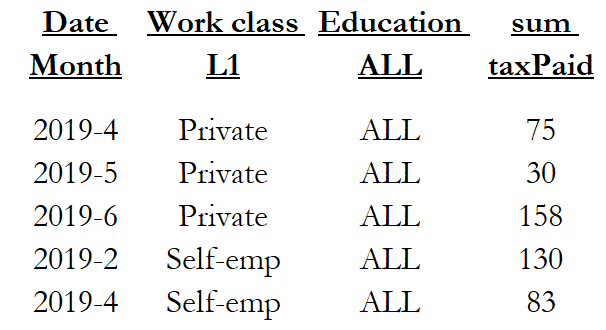}
      \caption{Results of queries $q^o$ and $q^n$}
      \label{fig:f-cont}
\end{figure}


\newpage
\section{Same-level Containment}\label{sec:SameLevelCont}
\subsection{Intuition: Checking for direct novelty via containment of cubes defined at the same levels}

Can we affirm that the cells of a certain (new) query, say $q^n$ are always a subset of another (possibly previously pre-computed) query, say $q$? 

A first precondition is that the two queries have exactly the same schema and the same aggregate functions applied to the same detailed measures, to even begin discussing a potential overlap. If this is not met, then no extra check is necessary.

\result{Assume now that the above requirement is met and the schemata and aggregations of two queries $q^n$ and $q$ are identical, and the only difference the two queries have is in their selection conditions. The \underline{decision} problem at hand is: given the query $q$ with condition $\phi$ and the query $q^n$ with condition $\phi^n$, both defined at the same schema $[L_1, \ldots, L_n$, $M_1, \ldots, M_m]$, without using the extent of the cells of the two queries, and by using only the  selection conditions and the common schema of the queries, can we compute whether the result of $q^n$ is a subset of the result of $q$?}

\result{In a similar vein, the respective \underline{enumeration} problem is: can we compute which cells of $q^n$ are already part of $q$, and which are not?}

In the rest of our deliberations, we  call a dimension a \emph{non-grouper}, when it is rolled-up to the level $ALL$ and thus, is practically excluded from the underlying aggregation of values. \emph{Groupers} on the other hand, are the levels of the dimensions that are not rolled-up to $ALL$, and thus, the query result produces coordinates other than $all$ for them. 

To give a concrete example: Assume a cube over $Product, Time, Geography$ with $Sales$ as measure. Assume that we have two queries both of which roll-up $Geography$ at level $ALL$, and report sales per month and product family. $Geography$ is a non-grouper, because it is rolled-up to the level $Geography.ALL$ and thus, is practically excluded from the underlying aggregation of values.The other two dimensions are groupers.

What can make the cells of the two queries be different? Potential reasons are:

\begin{itemize}
\item Different filters in non-grouper levels. Assume that one of the two queries applies the filter $Country$ = $Japan$ and other has the filter $true$. As another example, one query applies the filter $Country$ = $Japan$ and the other one the filter $Country$ = $China$. In either case, the cells of the result of the two cubes will have the same coordinates, but the values will be different, due to the different filters in the non-groupers. A side-effect of this is that we cannot even exploit the case where the old cube has the filter $Country$ $\in$ $\{China, Japan\}$ and the new one $Country$ = $China$, again, because the resulting cells have the same coordinates, but their aggregate values are different.
\item Problematic partial filters in grouper. Assume the above scenario, with the old query selecting months in [$January~2020$ .. $November~2020$] and the new query selecting $Day$ in $[1/1/2020$ .. $15/11/2020]$. The problem here is in November: both queries will roll up at the level of month, and thus will report the month November 2020, but the new query is filtering a subset of this month, and thus the aggregate cells will be different.
\item Different filters in grouper levels. Assume the value-set of the old cube is not a super-set of the value set of the new cube, for a grouper level. For example, again assume that both queries roll-up $Geography$ at level $ALL$, and report sales per month and product family, and the old query selects months in [$March~2020$ .. $September~2020$] and the new query selects months in [$September~2020$ .. $October~2020$].

\end{itemize}

Practically, we need to have identical selections for non-grouper levels, and "rollable" selection subsumption with respect to the grouping levels, for grouper levels. Theorem \ref{theor:Same-level-containment} formalizes the above observation. Before introducing the theorem, however, we need to introduce a few definitions.

\subsection{Terminology} 

\subsubsection{Groupers}
\begin{definition}
Given a query $q$ with a schema comprising a set of levels $[D_1.L_1,\ldots,D_n.L_n]$, over the respective dimensions: 
\begin{itemize}
    \item A dimension $D$ is a \textit{non-grouper}, when it's respective schema level is (rolled-up to) the level $ALL$.
    \item A dimension $D$ is a \textit{grouper}, when its respective level in the schema is not rolled-up to $ALL$.
\end{itemize}

\end{definition}
By extension of the terminology, we will also refer to the respective levels as groupers and non-groupers, too.

\begin{definition}
Given a multidimensional schema $S$ and a simple selection condition $\phi$ to which it participates, a dimension $D$ with a grouper level $D.L^{\gamma}$ at the schema level and a filter level $D.L^{\sigma}$ at $\phi$,  is characterized as follows:
\begin{itemize}
\item \emph{unbound}, if $D.L^{\sigma}$ = $D.ALL$ and the atom of $phi$ is $D.ALL$ $\in$ $\{D.all\}$ (equiv., $true$)
\item \emph{pinned grouper}, if both $D.L^{\gamma}$ and $D.L^{\sigma}$ $\neq$ $D.ALL$ \item \emph{pinned non-grouper}, if $D.L^{\gamma}$ = $D.ALL$ and $D.L^{\sigma}$ $\neq$ $D.ALL$
\end{itemize}

\end{definition}

\hrulefill
\begin{example}
Assume a query

\[q^o = \tuple{ \mathbf{DS}^{0},\ \phi^o,\ [Month, W.L_1, E.ALL, sumTaxPaid],\
[sum(TaxPaid)]\ 
}
\]

Then, Month and $W.L_1$ are groupers and Education is a non-grouper.\\

Concerning Education:
\begin{itemize}
    \item if the atom $E.ALL \in \{All\}$ is part of $\phi$ than the dimension is unbound, i.e., all the members of the education dimension are computed for the final result
    \item if an atom like $E.L_3 \in \{Post-secondary\}$ is part of $\phi$, then the dimension is a pinned non-grouper
\end{itemize}

Concerning Date:
\begin{itemize}
    \item if the atom $Date.ALL \in \{All\}$ is part of $\phi$ than the dimension is unbound
    \item if an atom like $D.Year \in \{2019,2020\}$ is part of $\phi$, then the dimension is a pinned grouper
\end{itemize}
\end{example}
\hrulefill

\subsubsection{Rollable dimensions, schemata and selection conditions}
\begin{definition}[Perfectly Rollable Dimension / Perfectly Rollable atom]
Assume a grouper level $D.L^{\gamma}$ and an atom $\alpha$:$D.L^{\sigma}$ $\in$ $V$, $V$ = $\{v_1, \ldots, v_k\}$.
Then, the dimension $D$ is \emph{perfectly rollable} with respect to the tuple ($L^{\gamma}$, $L^{\sigma}$, $V$), or, equivalently, $\alpha$ is \emph{perfectly rollable} with respect to $L^{\gamma}$, if one of the following  two conditions holds:\\
(a) $L^{\gamma}$ $\preceq$ $L^{\sigma}$ (which implies that every grouper value of $L^{\gamma}$ that qualifies is entirely included, as the selection condition is put at a higher level that the grouping, e.g., group by month, for year = 2020)\\
(b) $L^{\sigma}$ $\prec$ $L^{\gamma}$, and for each value $u_i$ $\in$ $dom(L^{\gamma})$: $u_i$ = $anc_{L^{\sigma}}^{L^{\gamma}}(v_i)$, \emph{all} $desc^{L^{\sigma}}_{L^{\gamma}}(u_i)$ $\in$ $V$ (i.e., the entire set of children of a grouper value $u$ is included in the computation of $u$).
\end{definition}

\result{The intuition behind perfectly rollable atoms, is that whenever a grouper value will appear at the result of a cube query, its entire set of descendants will have been included in the grouping.}

\hrulefill

A practical implication of perfect rollability is that this property propagates all the way to $L^0$, where all the detailed descendants of a value $u$ are qualified by the selection condition to participate in the computation of the aggregate value (Figure \ref{fig:perfectRoll}).
\result{Both conditions guarantee that, given a simple selection condition on a dimension and a grouper level, there are no grouper cells in the result of a cube that could be computed on the basis of only a subset of their detailed descendants, but rather, the entire range of descendant values are taken into consideration for their computation}.\\

\hrulefill

\begin{figure}
  \centering
    \includegraphics[width=\textwidth]{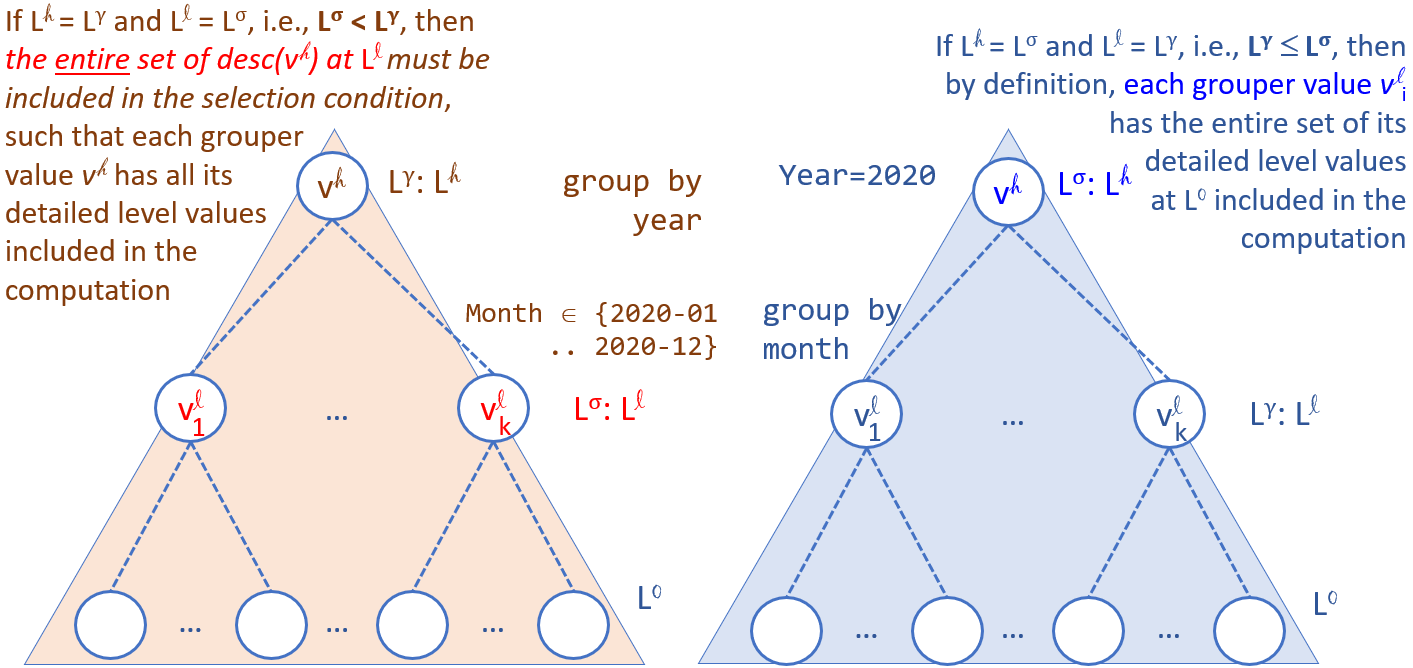}
      \caption{Perfect Rollability}
      \label{fig:perfectRoll}
\end{figure}

\begin{definition}[Perfectly Rollable Schema / Perfectly Rollable simple selection condition]
Assume a schema $S$: [$D_1.L_1$, $\ldots$, $D_n.L_n$] over a set of dimensions [$D_1$, $\ldots$, $D_n$] with each grouper level belonging to a different dimension and a simple selection condition $\phi$: $\bigwedge\limits_{i=1}^{n}$ $\alpha_i$, with each atom $\alpha$ of the form $D.L^{\sigma}$ $\in$ $V$, $V$ = $\{v_1, \ldots, v_k\}$, and exactly one atom per dimension.
Then, the schema $S$ is \emph{perfectly rollable} with respect to the tuple ($S$, $\phi$), or, equivalently, $\phi$ is \emph{perfectly rollable} with respect to $S$, if each atom $\alpha_i$ is perfectly rollable with respect to its respective grouper level $L_i$.
\end{definition}

The perfectly rollable condition is a "clean" characterization stating that if we group by a level $L$ on \emph{any} possible data set, then, the resulting grouper values of $L$ will be produced by the entire population of their descendants at lower levels (in fact, as far as the semantics are concerned: the most detailed one). 

However, perfect rollability is not the only useful situation that can occur in practice: consider a running-year summary of sales, which means that the entire sales of the year up to the current date are summed. It is quite possible that the current date is in the middle of the year, thus, when the year is summed, the entire population of its descendants is simply not there. This has been captured by the L-containment notion in \cite{DBLP:conf/caise/VassiliadisS00, VassiliadisPhdthesis} which is a broader concept than perfect rollability. \sideNote{Perfect Roll.~is a useful but not obligatory property}
However, perfect rollability is quite faster a check in all practical cases.

\newpage
\subsection{The decision problem of query containment for same-level queries}

Are the cells of a query result a subset of the sells of another query result? Can we decide whether this holds without actually ever executing the queries, just by their definition, and independently of the data stored in the database?

To answer these questions we introduce the following theorem. The theorem requires that the two queries have the same schema (otherwise there is no point to even discuss a subset relation). The two selection conditions must have the same filter for non-groupers, otherwise the filtering of the non-groupers is (a) different and (b) not observable at the cells, as the resulting cells will have $All$ as a coordinate, although internally there will be a filter posed to the members of the dimension. For the rest of the dimensions, they have to produce detailed areas where the one is a subset of the other, while both are perfectly rollable with respect to the common schema, such that the values of the coordinates of the result cells correspond to their entire detailed descendants. 

\begin{theorem}[Same-level-containment]\label{theor:Same-level-containment}
The query $q^b$ is a \result{same-level superset} of a query $q^n$ (equiv. \result{containing} a query $q^n$) if the following conditions hold:
\begin{enumerate}
    \item both queries have exactly the same underlying detailed cube $\mathbf{DS}$, exactly the same dimension levels in their schema and the same aggregate measures $agg_i(M_i^0)$, $i$ $\in$ 1 .. $m$ (implying a 1:1 mapping between their measures). To simplify notation, we will assume that the two queries have the same measure names, and thus, the same schema  $[L_1,~\ldots,~L_n,M_1,~\ldots,~M_m]$.
    \item both queries have simple selection conditions $\phi^b$ and $\phi^n$, respectively, with the following characteristics:
    \begin{enumerate}
    \item both queries have the same atoms for non-grouper dimensions
    \item grouper dimensions are (i) perfectly rollable with respect to the combination of their grouper and filter, and (ii) for each atom $\alpha_i^n$, with its detailed descendant being $\alpha_i^{n^0}$: $L^0~\in~V^{n^0}$, and the respective, $\alpha_i^b$, with its detailed descendant being $\alpha_i^{b^0}$: $L^0~\in~V^{b^0}$, the following condition holds: $V^{n^0}~\subseteq~V^{b^0}$.
    \end{enumerate}

\end{enumerate}
\end{theorem}

\begin{proof}
Without loss of generality, we assume a single measure $M$, in order to simplify notation.

\sideNote{Prove a broader signature}\textbf{Part A. Proving that $q^b$ has a broader signature than $q^n$}.
Since the two queries have the same groupers, the cells at the result of both queries will be at the same level. Now, we must ensure that for every cell in the results of $q^n$, say $c^n$, there exists a respective cell $c^b$ in $q^b$, with the same coordinates and exactly the same area at the most detailed level at $C^0$.

Due to the property 2, the following holds at the most detailed level

\begin{center}
$V_1^{0^n}$ $\times$ $\ldots$ $V_n^{0^n}$ $\subseteq$ $V_1^{0^b}$ $\times$ $\ldots$ $V_n^{0^b}$
\end{center}

This is due to (a) the identity of the value-sets of the non-grouper dimensions, and (b) the explicit requirement of condition 2b for groupers.

Then, when mapping the detailed values to the ancestors at the levels of the common schema of the two cubes, the value-sets produced will fulfill 
the respective condition at the grouper level:
\begin{center}
$V_1^{L_{1},n}$ $\times$ $\ldots$ $V_n^{L_{n},n}$ $\subseteq$ $V_1^{L_{1},b}$ $\times$ $\ldots$ $V_n^{L_{n},b}$
\end{center}
Thus, with respect to their coordinates,  the cells of the new cube are a subset of the cells of the broader cube.

\hrulefill

\textbf{Part A (Alternative)}. Another way to look at this is as follows:

\begin{table}[h!]
\begin{center}

\begin{tabular}{ c | c c }
  & groupers & non-groupers \\ 
  \hline\\
 pinned & $\phi$ & $\phi$ (must be same) \\  
 ~ & ~ & ~\\
  ~ & ~ & ~\\
 non-pinned & ALL=all & ALL=all (must be same)\\
 ~ & ~ & ~\\
 ~ & ~ & ~\\
\end{tabular}
\caption{Possibilities for query $q^b$}
\label{tab:directSup}
\end{center}
\end{table}

Assume the cell $c^\star$ in the result of $q^n$, $c^\star~\in~q^n.cells$, defined as the tuple $c^\star$ = [$c^\star_1$, $\ldots$, $c^\star_n$, $m$]. For each dimension $D_i$, the respective value $c^\star_i$ is produced as a result of an atom as filter at the most detailed level and the mapping to $L^n$ via an ancestor function:

\begin{enumerate}
    \item $all$, if $D_i$ is an unbound non-grouper dimension; since we have assumed identity for these dimensions, for each such dimension $D_i$, both cubes will have the same $atom_i^0$; 
    \item $all$, also in the case of a pinned non-grouper dimension, i.e., the grouping is done at level $ALL$, but there exists an atom filtering the dimension -- again we have assumed identity for these cases, so, at the detailed level, the two cubes will have the same atom $atom_i^0$ for each such dimension $D_i$ (observe that, here, perfect rollability does not hold, but this is acceptable by the theorem);
    \item a value $v$ in $dom(L_i)$, in the case of a grouper dimension; in this case, since for \emph{both} queries, dimension $D_i$ is rollable, and $V^{n^0}~\subseteq~V^{b^0}$, this means that the values produced for the dimension $D_i$ at $q^n$ will also include $v$ (and in fact, with exactly the same values $desc^{L^0}_{L_i}(v)$ at the detailed level $L^0$. 
    \begin{enumerate}
        \item For the case of unbound groupers, all the domain of the grouper level $L_i$ participates in the result; if $q^b$ has $D_i$ as an unbound grouper, no matter what $q^n$ has as a filtering atom, it is acceptable by definition (remember it is obligatorily perfectly rollable, thus, the common cells will be produced by the same detailed values).
        \item if $q^n$ has $D_i$ as an unbound grouper, then obligatorily by the theorem's condition, $q^b$ has $D_i$ as an unbound grouper, too -- otherwise condition 2b is violated.
    \end{enumerate}
    
\end{enumerate}
 
 \hrulefill

\sideNote{Prove same measures}
\textbf{Part B. Proving that aggregate cells have the same measure values}. Then, the only question that remains is: assume two cells $c^b$ and $c^n $belonging to $q^b.cells$ and $q^n.cells$, respectively and having the same coordinates. Do they have the same measure $m$? The question is reduced to whether the detailed area of a cell $c^b$ is exactly the same with the detailed area of a cell $c^n$ with exactly the same coordinates. 
\begin{itemize}
    \item Due to the fact that $V_1^{0^n}$ $\times$ $\ldots$ $V_n^{0^n}$ $\subseteq$ $V_1^{0^b}$ $\times$ $\ldots$ $V_n^{0^b}$, it is impossible for a detailed cell used for the computation of a cell of $q^n$, not to participate to the production of the respective cell of $q^b$ with exactly the same coordinates.
    \item Inversely, if the cell $c^b$ had even a single detailed cell $c^{b^0}$ not belonging to the respective detailed area of $c^n$, this would mean that $c^n$ would have to be produced by a violation of one of the two conditions of requirement (2): either (a) a non-grouper atom of $q^b$ was broader than the respective one of $q^n$, or, (b) if non-groupers were identical, a grouper dimension's atom $D.\alpha$ producing $c^n$ would not be perfectly rollable to the respective level $D.L$ of the schema (if it is perfectly rollable, then the respective detailed area is identical for the common value of $c^n$ and $c^b$ for $D.L$).

\end{itemize}
 Thus, for the cells with the same coordinates, the two queries have identical detailed areas, and therefore, the resulting measure is the same.


In summary, $q^n.cells$ $\subseteq$ $q^b.cells$, i.e., for each $c^\star$ in $q^n.cells$, there exists exactly the same cell in $q^b.cells$ (with the same coordinates and the same measure values).

\end{proof}

\hrulefill
\begin{example}
Take the two queries of the section~\ref{sec:fc-complex}, specifically:

\[q^o = 
\tuple{
    \mathbf{DS}^{0},\ \phi^o,\ [Month, W.L_1, E.ALL, sumTaxPaid],\ [sum(TaxPaid)]\ 
}
\]

having
\[\phi^o = Year \in \{2019,2020\} \wedge W.L_2 \in \{with-pay\}
\]

and
\[q^n = \tuple{ 
    \mathbf{DS}^{0},\ \phi^n,\ [Month, W.L_1, E.ALL, sumTaxPaid],\ [sum(TaxPaid)]\
}
\]

having
\[\phi^n = Year \in \{2019\} \wedge W.L_1 \in \{private, self-emp\}
\]

\begin{figure}[h]
  \centering
    \includegraphics[width=\textwidth]{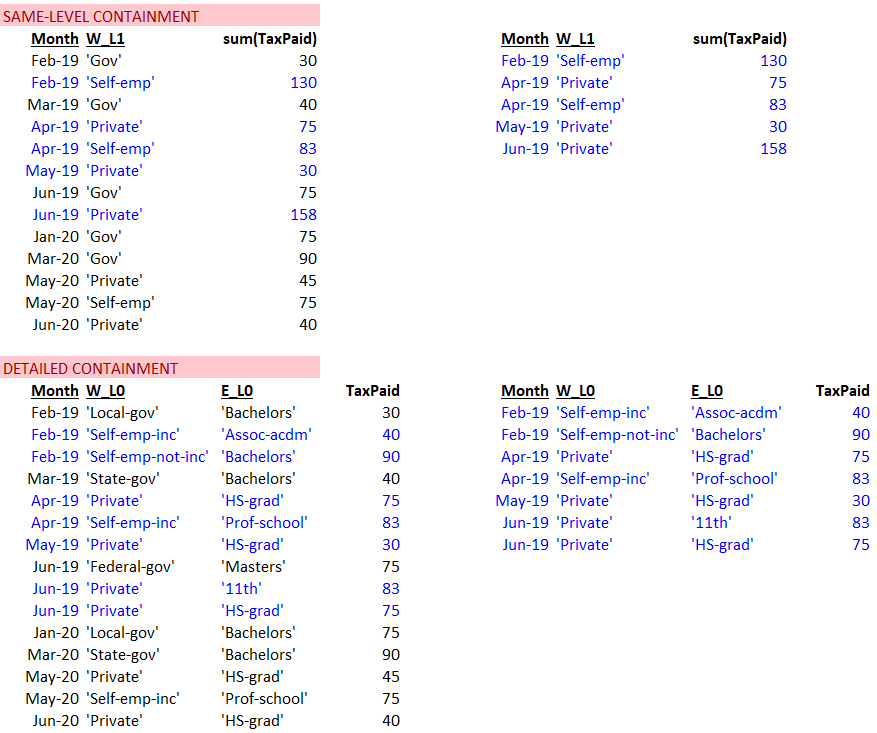}
      \caption{Containment of the two queries, $q^o$ (left) and $q^n$ (right). Same level containment (upper) and foundational containment (lower). }
      \label{fig:sml-containment}
\end{figure}

Then the Theorem holds for them. Specifically:
\begin{itemize}
    \item Both queries share the same schema. Thus condition 1 holds.
    \item Condition 2 also holds.
    \begin{itemize}
    \item Education is the only non-grouper, and the atoms in the two selection conditions are the same (i.e., $true$ or more accurately $E.ALL \in \{All\}$.
    \item The grouper dimension $Date$ is perfectly rollable for both queries (selection at the year level and grouping at the schema at month level).
    \item Concerning the dimension $Date$, the query $q^n$ has a set of detailed members that is a subset of the respective ones of $q^o$ -- specifically, $\{2019-01,\dots,2019-12\}$ $\subseteq$ $\{2019-01,\dots,2020-12\}$.
    \item The grouper dimension $Workclass$ is perfectly rollable for both queries, as in both queries, the grouper level is lower than the filter level at the respective atom. For $q^o$, $L_1 \preceq L_2$; for $q^n$, $L_1 \preceq L_1$.
    \item For the grouper dimension $Workclass$ the detailed members of $q^o$ are a superset of the respective ones of $q^n$ -- specifically, $\{private, not-inc, inc, federal, local, state\}$ $\supseteq$ $\{private, not-inc, inc\}$.
    \end{itemize}
\end{itemize}
Therefore, since both conditions hold, the theorem holds for these two queries.

\end{example}
\hrulefill


\subsection{The enumeration problem of query containment for same-level queries}

\result{Assume we have two queries, say an``old" query $q$ and a ``new'' $q^n$ under the exact same (i) underlying detailed cube $\mathbf{DS}^0$, (ii) schema $[L_1,~\ldots,~L_n,~M_1,~\ldots,~M_m]$, and (iii) aggregate measures $agg_i(M_i^0)$, $i$ $\in$ 1 .. $m$. Assume also that they both have simple selection conditions, albeit different: query $q$ with condition $\phi$ and the query $q^n$ with condition $\phi^n$, with the constraints of Theorem~\ref{theor:Same-level-containment} such that $q$ contains $q^n$. Can we compute which cells of the broader query $q$ are also part of the narrow query $q^n$, and which are not?
}

There are several alternatives for the problem: (a) a cell by cell \textit{after} the execution of both queries, or, (b) a comparison of the signatures of the two queries and the identification of coordinates that create a difference, \textit{without the need to execute the queries}. 


\subsubsection{Cell by cell}
The naive way to assess the enumeration problem is to compute the results and compare pairwise. 

\subsubsection{Answer based on Cartesian Produce Difference} 
The intuitive answer is to compute the Cartesian Product of result coordinates for both queries, say $q^{n^{+}}$ and $q^{+}$ and compute their difference $\delta^{+}$ = $q^{+}$ - $q^{n^{+}}$, containing the coordinates of the cells of the previous query $q$ not contained in the new query $q^n$. 
To the extent that the conditions of Theorem \ref{theor:Same-level-containment} are respected, the cells with the same coordinates will have exactly the same measures; therefore, the result of the set difference, $\delta^{+}$, will indicate exactly which cells are not already part of the results of the other query.
Algorithm~\ref{algo:EnumQueryContainment} performs this query signature comparison.

Basically, we need to invoke Algorithm~\ref{algo:QueryIntersectionSignatureDiff}  \textsc{\sc{ProduceCoveredAndNovelQueryCoordinates}} of section~\ref{sec:otherSign} but only if the pair of involved queries satisfy the conditions of Theorem \ref{theor:Same-level-containment}.

\begin{algorithm}[H]
\DontPrintSemicolon 
\KwIn{An old query $q$ containing a new query $q^n$, satisfying Theorem~\ref{theor:Same-level-containment} }
\KwOut{The subset of the coordinates of the old query $q$, say $q^{cov+}$, that are already part of the result of $q^{n}$, and its complement $q^{nov+}$}
\SetKwFunction{FCovNovSign}{ProduceCoveredAndNovelQueryCoordinates}
\Begin{
\If{ $q$ and  $q^n$  satisfy Theorem \ref{theor:Same-level-containment} } {
    \Return{$q^{cov+}$, $q^{nov+}$ $\gets$ \FCovNovSign($q^n$, $q$)}\; \tcp*{observe the order of param's}
    }
}
\caption{\sf{Containment Enumeration Of Common Cells Via Signature Comparison}}
\label{algo:EnumQueryContainment}
\end{algorithm}

\newpage
\section{Query Intersection}\label{sec:intersect}

Can we affirm that the cells of a certain query, say $q^1$ intersect with the result cells of another query, say $q^2$? 

Much like containment, a first precondition is that the two queries have exactly the same schema and the same aggregate functions applied to the same detailed measures, to even begin discussing a potential overlap. If this is not met, then no extra check is necessary.

\result{Assume now that the above requirement is met and the schemata and aggregations of two queries $q^1$ and $q^2$ are identical, and the only difference the two queries have is in their selection conditions. The \underline{decision} problem at hand is: given the query $q^1$ with condition $\phi^1$ and the query $q^2$ with condition $\phi^2$, both defined at the same schema $[L_1, \ldots, L_n$, $M_1, \ldots, M_m]$, can we compute whether the result of $q^1$ intersects with the result of $q^2$, without using the extent of the cells of the two queries, and by using only the  selection conditions and the common schema of the queries?}

\result{In a similar vein, the respective \underline{enumeration} problem is: can we compute which cells of $q^1$ are also part of $q^2$, and which are not?}

\subsection{The decision problem of query intersection for same-level queries}
\begin{theorem}[Same-level-intersection]\label{theor:Same-level-intersection}
The query $q^1$ has a \result{same-level intersection}, or simply, \result{intersects}, with a query $q^2$ if the following conditions hold:
\begin{enumerate}
    \item both queries have exactly the same underlying detailed cube $\mathbf{DS}$, the same dimension levels in their schema and the same aggregate measures $agg_i(M_i^0)$, $i$ $\in$ 1 .. $m$ (implying a 1:1 mapping between their measures). To simplify notation, we will assume that the two queries have the same measure names, and thus, the same schema  $[L_1,~\ldots,~L_n,M_1,~\ldots,~M_m]$.
    \item both queries have simple selection conditions $\phi^1$ and $\phi^2$, respectively, with the following characteristics:
    \begin{enumerate}
    \item both queries have the same atoms for non-grouper dimensions
    \item grouper dimensions are perfectly rollable with respect to the combination of their grouper and filter, 
    \item for every grouper dimension $D$, which is commonly grouped at level $D.L$ in both queries, and the pair of homologous atoms $\alpha^1$: $D.L^{1^\phi}$ $\in$ $V^1$, and the respective, $\alpha^2$: $D.L^{2^\phi}$ $\in$ $V^2$ (with $L^{2^\phi}$ and $L^{2^\phi}$ being potentially different), their signatures, transformed at the grouper level $D.L$ intersect $\alpha^{1@L^+}$ $\bigcap$ $\alpha^{2@L^+}$ $\neq$ $\emptyset$ \sideNote{equivalently $q^{1^+}$ intersects with $q^{2^+}$}
    \end{enumerate}

\end{enumerate}
\end{theorem}
 
\begin{proof}
Without loss of generality, we assume a single measure $M$, in order to simplify notation.
Since the two queries have the same groupers, the cells at the result of both queries will be at the same level. Now, we must ensure that  there exists at least a single cell in the results of $q^1$, that also exists in the results of $q^2$, i.e., in both results there is a cell with the same coordinates and exactly the same area at the most detailed level at $C^0$.

\sideNote{Let's prove the intersection of result coordinates is not empty}
Without loss of generality, let's assume the most restrictive case where each pair of homologous atoms has a single common grouper value when transformed to the grouper level. Specifically:

\begin{enumerate}
    \item Since $\alpha^1$: $D.L^{1^\phi}$ $\in$ $V^1$, let $V^{1^0}$ be the set produced by taking the union of values $desc_{L^{1^\phi}}^{L^0}(v)$, for each value $v$ $\in$ $V^1$, $V^{1^0}$ = $\bigcup\limits_{i=1}^{k}$ $desc_{L^{1^\phi}}^{L^0}(v_i)$.The set $V^{2^0}$ is produced similarly.
    \item When the values of $V^{1^0}$ are rolled-up to the ancestor values at the grouper level $L$, we get the value set $V^{1@L}$ which gives the grouper values for the query, for this respective dimension: $V^{1@L}$ = $\bigcup\limits_{i=1}^{k'}$ $anc^L_{L^0}(v)$, $v~\in~V^{1^0}$. Similarly, we produce $V^{2@L}$ from $V^{2^0}$.
    \item Let us assume, now, without loss of generality that $V^{1@L}$ and $V^{2@L}$ have only a single value in common, say $\gamma$ = $V^{1@L}$ $\bigcap$ $V^{2@L}$. Assume that this holds for all grouper dimensions (for non-grouper dimensions, due to the identity of the respective atoms, the produced sets $V^{i@L}$ are identical for the two queries). The extension of the sequel of this proof to an intersection including more values is straightforward.
\end{enumerate}

Naturally, depending on the relative positions of $L$ with $L^{1^\phi}$ and $L^{2^\phi}$ there are faster ways to produce the sets $V^{1@L}$ and $V^{2@L}$. The steps 1 and 2 of the above process produce a result independently of these positions, however, thus we omit the potential optimizations. 

Assume we have $m\star$ grouper dimensions. Given the above, if one considers the intersection of the Cartesian Products

\begin{center}
$V_1^{1@L_1}$ $\times$ $\ldots$ $\times$ $V_{m\star}^{1@L_m\star}$ 
$\bigcap$ 
$V_1^{2@L_1}$ $\times$ $\ldots$ $\times$ $V_{m\star}^{2@L_m\star}$
\end{center}
the result is nonempty and specifically: [$\gamma_1, \ldots, \gamma_{m\star}$].

By generalizing the Cartesian Product to include the non-grouper dimensions, too, the intersection of the respective value sets is non-empty (with the exception of the trivial case where a dimension produces an empty value set, in which case both queries have an empty result):
\begin{center}
$V_1^{1@L_1}$ $\times$ $\ldots$ $\times$ $V_n^{1@L_n}$ 
$\bigcap$ 
$V_1^{2@L_1}$ $\times$ $\ldots$ $\times$ $V_n^{2@L_n}$
$\neq~\emptyset$\\
(equivalently: $q^{1^+}$ $\bigcap$ $q^{2^+}$ $\neq~\emptyset$ )
\end{center}
and its projection to the grouper dimensions is [$\gamma_1, \ldots, \gamma_{m\star}$].

\sideNote{Prove same measure}
Then, the only question that remains is: assume two cells $c^1$ and $c^2$ belonging to $q^1.cells$ and $q^2.cells$, respectively and having the same coordinates. Is it the same cell? I.e., do they have the same aggregate measure $M$? To the extent that both queries work with the same measure and aggregate function, the question is reduced to whether the detailed area of a cell $c^1$ is exactly the same with the detailed area of a cell $c^2$. 

Let's assume the cells with aggregate coordinates [$\gamma_1, \ldots, \gamma_n$]. Lets assume, without loss of generality that the first $m\star$ dimensions are the grouper ones and the rest are the non-groupers. Then, the detailed proxy for both cells $c^1$ and $c^2$ is produced by the Cartesian Product $\gamma_1^0 \times \ldots \times \gamma_n^0$ which is computed \emph{exactly} via the expression:
\begin{center}
$desc_{L_1}^{L_1^0}(\gamma_1), \times \ldots \times desc_{L_m\star}^{L_{m\star}^0}(\gamma_{m\star}) \times \Gamma^{(m\star+1)^0} \times \ldots \times \Gamma^{n^0}$ 
\end{center}

with the sets $\Gamma^{i}$ of the non-grouper dimensions being identical for both queries (remember that non-grouper filters are identical). 
 
Is it possible that there exists a detailed cell participating in the production of say $c_1$ and not in the production of $c_2$? The answer is negative: the grouper coordinates are produced by the entire set of descendants (and only them) --otherwise grouper dimensions are not rollable-- and the non-grouper dimensions have identical filters.

Thus, for the cells with the same coordinates, the two queries have identical detailed signatures (and as a result, detailed areas too), and, therefore, the resulting measure is the same. Consequently, there exists at least one common cell in the result of the two queries.

\end{proof}

\hrulefill
\begin{example}
Assume the following two queries with the same schema and different selection conditions.

\[q^o = \tuple{ \mathbf{DS}^{0},\ \phi^o,\ [D.Month, W.L_1, E.ALL, sumTaxPaid],\
[sum(TaxPaid)]\ }
\]

having
\[\phi^o = Year \in \{2019,2020\} \wedge W.L_2 \in \{with-pay\}
\]

and
\[q^n = \tuple{ \mathbf{DS}^{0},\ \phi^n,\ [D.Month, W.L_1, E.ALL, sumTaxPaid],\
[sum(TaxPaid)]\ }
\]

having
\[\phi^n = Year \in \{2018,2019\} \wedge W.ALL \in \{All\}
\]

Then, we can see that all the conditions of the theorem are held:
\begin{itemize}
    \item both queries have the same schema;
    \item all the atoms of the two selection conditions are in the form requested by the query (both queries imply an atom of the form $E.ALL \in \{All\}$ too);
    \item both queries have the same atom for the non-grouper dimension $Education$;
    \item all grouper dimensions are perfectly rollable with respect to their respective atoms, as the atoms are all defined at higher levels than their respective grouper levels;
    \item the detailed proxies of the query atoms intersect as the following table shows; this holds both for the $Date$ atoms $\alpha^o_D$ and $\alpha^n_D$ and for the $Workclass$ atoms, $\alpha^1_W$ and $\alpha^2_W$, respectively
\end{itemize}

\begin{center}
\begin{tabular}{ l l }
 $\alpha^{o^{0+}}_D$: $\{2019-01, \ldots, 2020-12\}$ & $\alpha^{o^{0+}}_W$: $\{priv., nonInc, inc, fed., loc., st.\}$  \\ 
 $\alpha^{n^{0+}}_D$: $\{2018-01, \dots, 2019-12\}$ & $\alpha^{n^{0+}}_W$: $\{priv., nonInc, inc, fed., loc., st.,WOPay\}$  \\   
\end{tabular}
\end{center}

Therefore, all the conditions for the Theorem hold, and thus, the two queries are guaranteed to have common cells.

\begin{figure}[h]
  \centering
    \includegraphics[width=0.9\textwidth]{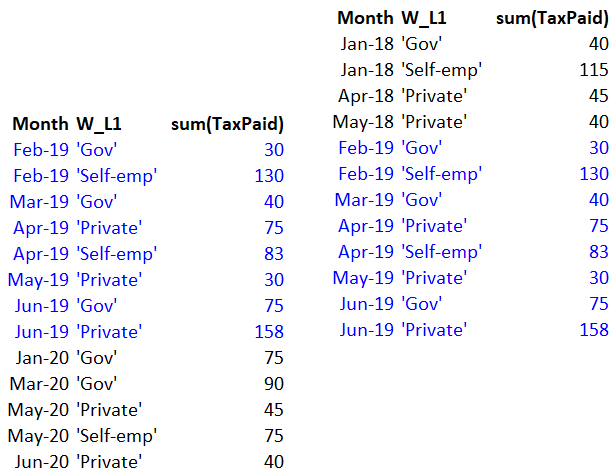}
      \caption{Intersection of the two queries }
      \label{fig:sml-intersect}
\end{figure}

\end{example}
\hrulefill

\subsection{Enumeration of the common cells of two queries} 
Assuming we contrast a query $q$ against a query $q^{\star}$, can we compute which cells of $q$ are contained in the other query?

The intuitive answer is to compute the Cartesian Product of result coordinates for both queries, say $q^{\star^{+}}$ and $q^{+}$ and compute their difference $\delta^{+}$ = $q^{+}$ - $q^{\star^{+}}$, containing the coordinates of the cells of the query $q$ not contained in the query $q^{\star}$. 
To the extent that the conditions of Theorem \ref{theor:Same-level-intersection} are respected, the cells with the same coordinates will have exactly the same measures; therefore, the result of the set difference, $\delta^{+}$, will indicate exactly which cells are not already part of the other query's results.

Basically, we need to invoke Algorithm~\ref{algo:QueryIntersectionSignatureDiff}  \textsc{\sc{ProduceCoveredAndNovelQueryCoordinates}} but only if the pair of involved queries satisfy the conditions of Theorem \ref{theor:Same-level-intersection}.

\begin{algorithm}[H]
\DontPrintSemicolon 
\KwIn{A query $q$ and a benchmark query $q^\star$, satisfying Theorem~\ref{theor:Same-level-intersection} }
\KwOut{The subset of the coordinates of $q$, say $q^{cov+}$ that are already part of the result of $q^{\star}$, and its complement $q^{nov+}$}
\SetKwFunction{FCovNovSign}{ProduceCoveredAndNovelQueryCoordinates}
\Begin{
\If{ $q$ and  $q^\star$  satisfy Theorem~\ref{theor:Same-level-intersection} } {
    \Return{$q^{cov+}$, $q^{nov+}$ $\gets$ \FCovNovSign($q$,  $q^\star$)}\;
    }
}
\caption{\sf{Enumerate Common Cells Via Signature Comparison}}
\label{algo:EnumQueryIntersectionSignatureDiff}
\end{algorithm}

A very similar check can be made at the extensional level, if the results of the queries are already available. The check on whether two cells are the same can depend only on the coordinates if the Theorem~\ref{theor:Same-level-intersection} holds.

\begin{algorithm}[H]
\DontPrintSemicolon 
\KwIn{A query $q$ and a benchmark query $q^\star$, satisfying Theorem~\ref{theor:Same-level-intersection} }
\KwOut{The subset of the result cells of $q$, say $q^{cov}$ that are already part of the result of $q^{\star}$, and its complement $q^{nov}$}
\Begin{
    \If{results are not available}{
        produce $q.result$ and $q^{\star}.result$ \;
    }
\tcp*{all identity checks can be coordinate-based}
$q^{cov}$ $\gets$ $q.result$~$\bigcap$~$q^{\star}.result$  \;
$q^{nov}$ $\gets$ $q.result$ - $q^{\star}.result$ \;

\Return{$q^{cov}$, $q^{nov}$ }\;
}
\caption{\sf{Enumerate Common Cells Via Result Comparison}}
\label{algo:EnumQueryIntersectionCellDiff}
\end{algorithm}

\subsection{Checking a query for containment over a query set}
Assume we have a list of queries previously issued by a user in the context of his history, say $Q$ = $\{ q_1, \ldots, q_k\}$. Then, a new query follows in the session, which for simplicity, we call simply $q^n$. We also assume that the underlying cube has not changed values in the context of the session. \emph{What we would like to be able to check is the subset of cells of $q^n$ that have previously been obtained via the previous queries, without checking the results of the queries of $Q$. Ideally, this should be achievable before issuing the query $q^n$ and obtaining its results}.

\subsubsection{Containment Syntactic check: a coarse approximation}

The simplest possible check is to see if a query is already contained in the list of previous queries. Let's assume that we test $q$ against each of the existing $q_i$ queries individually. We return $true$ if the definition of the query $q$ is found in $Q$ and $false$ otherwise.

\subsubsection{Intersection Syntactic Check}
Let's take the new query $q^n$ against an existing query $q$, $q$ $\in$ $Q$.  What if we can compare two queries pairwise and we can label the subset of the multidimensional space of $q^n$ already covered by $q$, as $visited$? Then, we can test the new query $q^n$ against each of the queries in $Q$ and stop once we exhaust them all, or, we have covered the entire space of $q$, whichever comes first. Algorithm~\ref{algo:EnumQueryIntersectionSignatureDiff} is the basis for this check, progressively updating the set of covered and uncovered cells.

\begin{algorithm}
\DontPrintSemicolon 
\KwIn{A list of queries $Q$ = $\{q_1, q_2, \ldots, q_n\}$ and a new query $q$}
\KwOut{The subset of the coordinates of $q$, say $q^{cov}$ that are already part of the results of $Q$ queries, and its complement $q^{nov}$}
\KwData{Interim query set $Q^\star$, interim set of coordinates $T$}
\Begin{
$q^{cov} \gets \emptyset $\;
$q^{nov} \gets q^{+} $ \tcp*{all the coordinates of q}
$Q^\star \gets$ all queries of $Q$ satisfying Theorem~\ref{theor:Same-level-intersection}\;
\ForAll {$q^{\star} \in Q^{\star}$}{
    produce $q^{\star^{+}}$ and intersect it with $q^{+}$, 
    $T$ = $q^{\star^{+}}~\bigcap~q^{+}$\;
    $q^{cov} \gets q^{cov}~\bigcup~T$; $q^{nov} \gets q^{nov}~-~T$ \;
}
\Return{$q^{cov}$, $q^{nov}$ }
}
\caption{\sf{Compute Partial Immediate (Same-level) Cube Coverage }}
\label{algo:partialImmediateCoverage}
\end{algorithm}

Algorithm~\ref{algo:partialImmediateCoverage} isolates all queries that are candidate for a same-level containment computation, and progressively identifies, which of their cells are also part of the new query's, $q$, result.


\newpage
\section{Query Distance}\label{sec:queryDistance}

\result{How similar are two queries?}
Fundamentally, the distance of two queries is the weighted sum of the distances of their components (see \cite{DBLP:conf/icde/BaikousiRV11}, \cite{DBLP:journals/dss/GolfarelliT14}, \cite{DBLP:journals/kais/AligonGMRT14}). To support our discussion in the sequel we assume two queries \textit{over the same data set} in a multidimensional space of $n$ dimensions.

\[q^a = \tuple{ \mathbf{DS}^{0},\ \phi^a,\ [L_1^a,\ldots,L_n^a, M_1^a,\ldots,M_{m^a}^a],\ [agg^a_1(M^{a^0}_1),\ldots,agg^a_m(M^{a^0}_{m^a})]\ }\]
and
\[q^b = \tuple{\mathbf{DS}^{0},\ \phi^b,\ [L_1^b,\ldots,L_n^b, M_1^b,\ldots,M_{m^b}^b],\ [agg^b_1(M^{b^0}_1),\ldots,agg^b_m(M^{b^0}_{m^b})]\ }\]

As in all other cases, we will assume that the selection conditions are simple selection conditions, including $n$ atoms, each of the form $D.L~\in~\{v_1, \ldots, v_k\}$.

The distance of the two queries is expressed by the formula

\[ \delta(q^a, q^b) = w^\phi\delta^\phi(q^a, q^b) + w^L\delta^L(q^a, q^b) + w^M\delta^M(q^a, q^b) ,\]  

\noindent such that the sum of the weights $w^i$ adds up to 1. \\

\textbf{Weights}. We follow \cite{DBLP:journals/kais/AligonGMRT14} and recommend the following weights: $w^\phi$: 0.5, $w^L$: 0.35, $w^M$: 0.15.\\

\textbf{Distance of atoms and selection conditions}. Consider two atoms $\alpha^a: D.L^a~\in~\{v_1^a, \ldots, v_k^a\}$ and $\alpha^b: D.L^b~\in~\{v_1^b, \ldots, v_k^b\}$ over levels of the same dimension. Let $V^{a^0}$ be the detailed proxy of the value-set of the first atom and $V^{b^0}$ the respective value set for the second atom. Thus, the detailed proxies of the two atoms are $\alpha^{a^0}: D.L^0~\in~V^{a^0}$ and $\alpha^{b^0}: D.L^0~\in~V^{b^0}$. The similarity of the two atoms is then the Jaccard distance of the detailed proxies of the value sets and their distance is its complement.

\[ \delta^\phi(\alpha^a, \alpha^b) = 1- \frac{ V^{a^0} \bigcap V^{b^0} }{V^{a^0} \bigcup V^{b^0}}  \]

The above formula holds for any lattice of levels. For the special (but very frequent) case of (i) a chain (total order) of levels, which means that it is necessary from the definition that the level of one atom is a descendant of the level of the other atom -- without loss of generality say $L^a \preceq L^b$, and, (ii) dicing atoms where both atoms are of the form $L~=~v$, we can decide that if $v^b~\neq~anc^{L^b}_{L^a}(v^a)$, then the distance of the two atoms is 1.

Then, the distance of two simple selection conditions, each having a single atom per dimension, is simply:

\[ \delta^\phi(q^a, q^b) = \frac{1}{n} \sum_{i=1}^{n} \delta^\phi(\alpha^a_i, \alpha^b_i) \]

\textbf{Distance of levels}. To define the distance of the two schemata, first we need to define the distance of two levels in the same dimension. This can be obtained with a variety of metrics, however, we choose to keep a simple definition.

Assume a dimension which is a simple chain (total order) of levels, starting at the lowest possible level $L^0$ and ending at the highest possible $ALL$. We provide the following definitions.

\[ height(L) = \text{ the number of edges crossed from } L^0 \text{ to reach } L \]

\[ \delta^L(L_1, L_2) = \frac{|height(L_1) - height(L_2)|} {height(ALL)}  \]\\

The generalization of this is that the denominator can take the value $max~height$ (or even $max~height$ - $min~height$ if heights do not start from zero). 

Also, for a lattice of levels, with $L^0$ as the lowest and $ALL$ as the highest level, the definitions become:

\[ height(L) = \text{ the number of edges crossed for the maximum path from } L_0 \text{ to } L \]

\[ \delta^L(L_1, L_2) = \frac{ \text{number of edges crossed from } L_1 \text{ to } L_2 \text{ via the min path among their ancestors } } 
   {\text{number of edges of the maximum ancestor path for any two nodes in the lattice}}  \]\\  

\noindent where the main idea is that the involved paths, apart from the start and end nodes comprise only their ancestors.\\

Then, the distance of two schemata, each having a single level per dimension, is simply:

\[ \delta^L(q^a, q^b) = \frac{1}{n} \sum_{i=1}^{n} \delta^L(L^a_i, L^b_i) \]

\textbf{Distance of aggregate measures}. The formulae for the production of the aggregate measures over the detailed ones are fairly easy to define, as they are based on the identity of aggregate functions and detailed measures. However, since the sets of measures of the two queries can be of arbitrary cardinality and in arbitrary order, we need to handle this too in the measuring of the distance.

Before declaring the distance of aggregate measures, we need to accurate the mapping of \textit{homologous} aggregate measures with the same aggregate function and detailed measure, between the two queries. Specifically, 

\[ mapM(M^{x}_i) = \begin{cases}
        M^{y}_j & \text{if } M^{x^0}_i \equiv M^{y^0}_j \text{ and } agg^x_i \equiv agg^y_j \\   
        null & \text{otherwise}
     \end{cases}
\]

Thus, we map each aggregate measure to a homologous one in the other query, if such a measure exists, or null, otherwise.

\[ \delta^M(q^a, q^b) = \frac{{m^a}}{{m^a}+{m^b}} \sum_{i=1}^{{m^a}} isNull(mapM(M^{a}_i))\ 
+ \frac{{m^b}}{{m^a}+{m^b}} \sum_{i=1}^{{m^b}} isNull(mapM(M^{b}_i))
\]

A possible way to implement this is to insert the $agg^\star_i(M^{\star^0}_i)$ of both queries in the same hash map $<agg^\star_i(M^{\star^0}_i), counter>$ and count the occurrences of each $agg^\star_i(M^{\star^0}_i)$: those who are equal to two produce a distance of zero, and the rest a distance of one.\\



\newpage
\section{Cube Usability: computing a cube from another cube whose result is available}\label{sec:cubeUsa}

What if we are not interested in checking whether the contents of the new query $q^n$ are contained with the results of a previous query $q$, but derivable from it? Practically, this means that the requirement that the schema is the same is violated.

For example, what if, all else being equal, the new query groups data by year and the previous query groups them by month? In this case, there is containment (in fact:  equivalence) at the most detailed level, but due to the fact that the condition on schema identity is violated, the two queries are not same-level comparable.

\result{The \underline{decision} problem for cube usability concerns the determination of whether a previous query $q^b$ can be exploited to compute the new one $q^n$, by transforming its results. Apart from deciding whether this is possible, however, we also need \underline{an algorithm} for performing the computation.} 

Before that, however, we will need a small digression, to introduce distributive aggregate functions.

\subsection{Distributive aggregate functions}
The main idea with any aggregate function, like e.g., sum or max, is that it takes a bag of a values (we assume real values) as input and it produces a result value as output. Think, for example, of the case when we are computing the measure of a specific cell of a query result, by applying the respective $agg$ function to the detailed cells that pertain to it. When it comes to distributive aggregate functions, the idea is that, alternatively to computing the  result over the original bag of values, it is possible to use as input pre-computed summaries which have been computed over a partition of the original bag to disjoint subsets of it. Think, for example, of the case that we would like to compute a sum for a specific year, and for some reason, we already possess the sum per month of this year. A distributive aggregate function comes with a guarantee than whenever the partitioning reflects an equivalence relation (all members of the original domain have found a group in the intermediate computation, and the groups are disjoint), then, it is possible to compute the requested result, not over the original domain, but over the (hopefully much smaller) domain of the already available summary. See Figure~\ref{fig:distributive} for an example.\\ 

Formally, assume a bag of measure values $M$ = $\{m_1, \dots, m_k \}$. These can be measures of detailed cells, or of cells at any aggregation level, that are going to be further aggregated. Observe that we have a bag and not a set of such values (as two cells can possibly have the same measure value). We assume that the members of $M$ are members of $\mathbb{R}$.

Assume also an aggregation function $agg$: $2^\mathbb{R}$ $\rightarrow$ $\mathbb{R}$ applied to $M$ and producing a single value $v$ = $agg(M)$. 

Assume, finally, an equivalence relation of $M$, reflecting a disjoint partitioning of $M$ to subgroups. Let $M^{g^U}$ be such a disjoint partition $M^{g^U}$ = $\{\mathbf{g_1} \ldots \mathbf{g_x}\}$, s.t., $\mathbf{g_i} \cap \mathbf{g_j} = \emptyset, i \neq j$ and $M = \cup_{i=1}^{x}\mathbf{g}_{i}$. We denote the bag of values $M^U$ as the result of applying $agg$ to each member of $M^{g^U}$.

\textbf{Distributive functions and their facilitators}. Then, $agg$ is a \emph{distributive function} if there exists another (possibly, but not necessarily different) \emph{facilitator} aggregation function $agg^{F}$, s.t., for any $M^{g^U}$  being an equivalence relation of  $M$,

\[
agg(M) = agg^{F}(M^U) = agg^{F}(agg(M^{g^U}(M))) 
\]

Typical examples of such distributive functions are $sum, max, min, count$ with their facilitators $agg^{F}$ being identical to $agg$ except for $count$, where $agg'$ = $sum$ (i.e., $count(M) = sum(M^U), M^U = count(M^{g^U}(M))$).

We denote the facilitator of an aggregate function $agg$ as $agg^{F}$.

\begin{figure}[tb]
  \centering
    \includegraphics[width=0.9\textwidth]{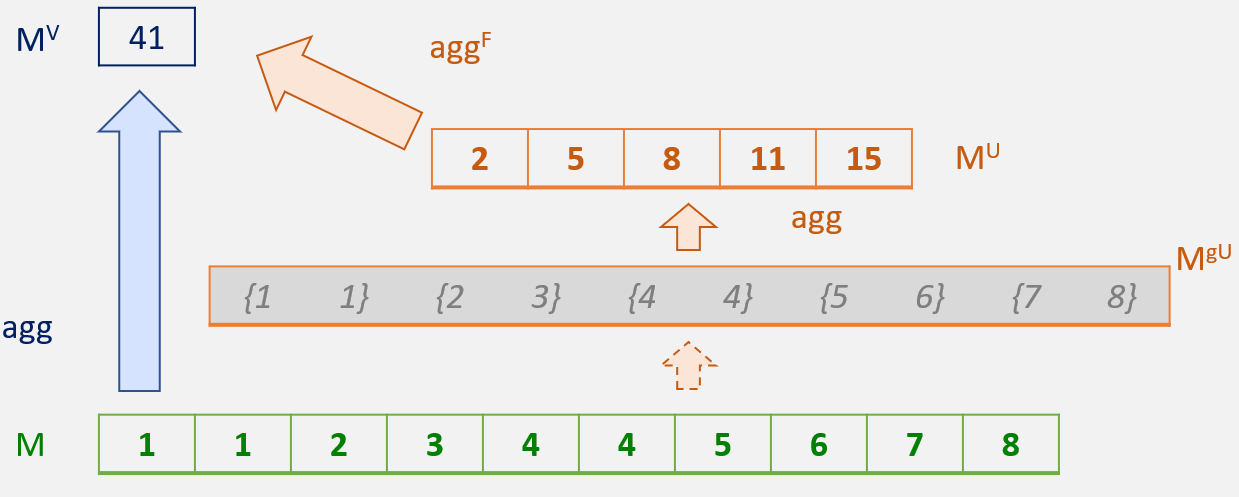}
      \caption{An example of the distributive function \textsf{sum}}
      \label{fig:distributive}
\end{figure}

\subsection{Preliminaries}
To support our discussion, in the sequel, we assume two queries, to which we refer to as $q^b$ (with the hidden implication of "broad" in terms of selection condition, "below" in terms of the level of the grouping, and, "before" in terms of creation) and $q^n$ (with the hidden implication of "narrow" in terms of selection condition, "not lower" in terms of the level of the grouping, and, "new" in terms of creation). As in all other cases, we will assume that the selection conditions are simple selection conditions, including $n$ atoms, each of the form $D.L~\in~\{v_1, \ldots, v_k\}$. We also assume that all aggregate functions are distributive. To simplify the presentation even more, we assume that there is a one-to-one mapping between the measures of the two queries and the respective aggregate functions that produce them (thus, the two queries differ only with respect to the levels of their schema and their selection conditions).

\subsection{Cube usability theorem and rewriting algorithm}

\begin{theorem}[Cube Usability]\label{theor:cubeUsability}
Assume the following two queries:

\[q^n = \tuple{ \mathbf{DS}^{0},\ \phi^n,\ [L_1^n,\ldots,L_n^n, M_1,\ldots,M_m],\ [agg_1(M_1^0),\ldots,agg_m(M_{m}^0)]\ }\]
and
\[q^b = \tuple{\mathbf{DS}^{0},\ \phi^b,\ [L_1^b,\ldots,L_n^b, M_1,\ldots,M_m],\ [agg_1(M_1^0),\ldots,agg_m(M_{m}^0)] }\]

The query $q^b$ is \result{usable for computing}, or simply, \result{usable for} query $q^n$, meaning that Algorithm~\ref{algo:CubeUsabilityRewriting} correctly computes $q^n.cells$ from $q^b.cells$, if the following conditions hold:
\begin{enumerate}
    \item both queries have exactly the same underlying detailed cube $\mathbf{DS}$, 
    \item both queries have exactly the same dimensions in their schema and the same aggregate measures $agg_i(M_i^0)$, $i$ $\in$ 1 .. $m$ (implying a 1:1 mapping between their measures), with all $agg_i$ belonging to a set of known distributive functions. To simplify notation, we will assume that the two queries have the same measure names, 
    \item both queries have exactly one atom per dimension in their selection condition, of the form $D.L~\in~\{v_1, \ldots, v_k\}$ and selection conditions are conjunctions of such atoms,
    \item both queries have schemata that are perfectly rollable with respect to their selection conditions, which means that grouper levels are perfectly rollable with respect to the respective atom of their dimension, 
    \begin{itemize}
        \item (for convenience) for both queries, for all dimensions $D$ having $D.L^g$ as a grouper level and $D.L^{\phi}$ as the level involved in the selection condition's atom for $D$, we assume $D.L^{g}$ $\preceq$ $D.L^{\phi}$, i.e., the selection condition is defined at a higher level than the grouping \rem{see next subsection for the other case too}
    \end{itemize}
    \item all schema levels of query $q^n$ are ancestors (i.e, equal or higher) of the respective levels of $q^b$, i.e., $D.L^{b}$ $\preceq$ $D.L^{n}$, for all dimensions $D$, and,
    \item for every atom of $\phi^n$, say $\alpha^n$, if (i) we obtain $\alpha^{n@L^{b}}$ (i.e., its detailed equivalent at the respective schema level of the previous query $q^b$, $L^b$) to which we simply refer as $\alpha^{n@b}$, and, (ii) compute its signature $\alpha^{{n@b}^{+}}$, then (iii) this signature is a subset of the grouper domain of the respective dimension at $q^b$ (which involves the respective atom $\alpha^b$ and the grouper level $L^b$), i.e., $\alpha^{{n@b}^{+}}$ $\subseteq$ $gdom(\alpha^b, L^{b})$.
\end{enumerate}
\end{theorem}

\begin{algorithm}[thb]
\DontPrintSemicolon 
\KwIn{A new query expression $q^n$ and a previously computed query $q^b$ along with its result $q^b.cells$ }
\KwOut{The result of $q^n$, $q^n.cells$}
\Begin{
$q^n.cells$ $\gets$ compute $q^{n^+}$ and for every coordinate, create a new cell with all measures initialized to $\varnothing$ \;

\If {$q^{b}$ and $q^{n}$ satisfy  all conditions of Theorem~\ref{theor:cubeUsability}}{
    \ForAll {dimensions $D_i$}{
        $\alpha^{n@b}_i$ $\gets$ the transformed atom of the new query at the schema grouper level $L^b_i$ of $q^b$ \;
    }
    $\phi^{n@b}$ = $\wedge$ $\alpha^{n@b}_i$ \;
    $q^{n@b}.cells$ $\gets$ apply $\phi^{n@b}$ to $q^b.cells$ \;
    $q^{n^G}$ = group the cells of $q^{n@b}.cells$ according to $q^{n^+}$ \;
    \ForAll {measures $M_j$}{
        $q^n.cells.M_j$ $\gets$ apply $agg^F_j$ to the j-th measure of the members of the groups of $q^{n^G}$ \;
    }
}

\Return{$q^n.cells$}
}

\caption{\sf{Answer Cube Query from a Pre-Existing Query Result}}
\label{algo:CubeUsabilityRewriting}
\end{algorithm} 
 
\begin{proof}
We can prove that Algorithm~\ref{algo:CubeUsabilityRewriting} correctly computes the result of $q^n$. We will do this by going through the steps of the algorithm, and discussing their properties.
First we compute an empty result set with "illegal", null measure values to serve as a placeholder for the result.
Then we check for the conditions of Theorem~\ref{theor:cubeUsability} and if they do not hold, we return the aforementioned illegal result.
The essence of the proof begins once we are inside the body of the \textsf{if} statement in Line 3.
By now, for every dimension $D$ we have the following conditions satisfied:
\begin{itemize}
    \item $D.L^{g^b}$ $\preceq$ $D.L^{\phi^b}$
    \item $D.L^{g^n}$ $\preceq$ $D.L^{\phi^n}$
    \item $D.L^{g^b}$ $\preceq$ $D.L^{g^n}$
    \item and therefore, $D.L^{g^b}$ $\preceq$ $D.L^{\phi^n}$ 
\end{itemize}

This means that every atom of the new selection condition $\phi^n$ is expressed over a level that is an ancestor (again: equal or higher) than the respective level of the schema of the previous query result. Therefore, the selection condition of the new query is directly translatable and applicable to the pre-computed cells of the old cube. Therefore, we are entitled to compute $\alpha^{n@b}$ which is the respective detailed equivalent of $\alpha^{n}$  at the schema level $D.L^{g^{b}}$. The conjunction of all these atoms produces a new selection condition $\phi^{n@b}$ which is the detailed equivalent of the selection condition of $q^n$ expressed at the levels of the schema of $q^b$.

It is very important to note the above property: as the selection condition $\phi^{n@b}$ is the detailed equivalent of $\phi^{n}$ it defines \underline{exactly} the same detailed area of coordinates at the most detailed level of the data set $DS$. Clearly, applying it to $q^b$ cannot involve more detailed cells than the ones of the query $q^n$. However, could it be possible that it involves less than the ones desired by $q^n$?

The trick is in the combination of perfect rollability and property 6 of the Theorem. Since both cubes are perfectly rollable with respect to their selection conditions, this means that every (aggregate) cell that appears in the result of any of the two cubes is produced by the entirety of the detailed cells that pertain to its (aggregate) coordinates. Now observe that Condition 6 imposes that the signature of each atom of the transformed selection condition $\phi^{n@b}$ is a subset of the grouper domain of the old cube: this means, that the cells of the previous cube are (a) a superset of the respective ones of the new cube (at the lower-level schema of the old $q^b$) and (b) computed from the entirety of the detailed cells that have to produce them. Therefore, it is not possible that a certain part of the detailed equivalent of $q^n$ is not covered by $q^b$. 

So, when at Line 8 of the algorithm $\phi^{n@b}$ is applied to the cells of $q^b$, the detailed area covered by the result of this selection as the data set $q^{n@b}$ is \underline{exactly} the area of the new query, correctly aggregated at the levels of the previous query, without any omissions or surplus.

Then, all we need to do is roll-up the new data set correctly to produce $q^{n}$. To this end, we exploit the monotonicity of the dimensions and the facilitators of the distributed aggregate functions. The step in Line 9 constructs an equivalence relation: since all the cells of $q^{n@b}$ are at lower or equal levels with respect to $q^{n}$, we can group them according to the values of $q^{n^{+}}$. The monotonicity of the ancestor functions guarantees a correct mapping. 
Is it possible that a cell of $q^{n@b}$ does not find an ancestor signature at $q^{n^{+}}$? No, because $\phi^{n@b}$ produces exactly the same area of detailed cells with $\phi^{n}$, and thus no matter which the levels involved are, for every $\gamma~\in~\phi^{n@b^{+}}$ there is obligatorily a $\gamma^n~\in~\phi^{n^{+}}$ which is at an ancestor level for all dimensions. Is it possible that a cell is missing? No, for the exact same reason.

The computation of the measure values (Lines 10 - 12) is possible due to the property of the facilitators of the distributive functions: we have an equivalence relation, with classes dictated by the signatures of the query result, and the facilitator function invoked for each measure. By the definition of distributive functions, the result of each cell is produced exactly by the cells of $q^{n@b}$ that pertain to it, via the application of the appropriate facilitator aggregate function. Therefore, the result is the same with $q^n.cells$.

\end{proof}

\textbf{Example}. An example of cube usability is show in Figure~\ref{fig:exUsable}, diving in the details of a subset of Figure~\ref{fig:ex}. The new query $q3$ is derivable from the previous query  $q2$ immediately. The two tables in the figure show why and how this is done. The table at the middle of the Figure allows us to easily check the conditions 1-5 of the Theorem. The table at the lower part of the figure helps checking the 6th condition of the Theorem. Following is the detailed check of all conditions of the Theorem:
\begin{enumerate}
    \item Both queries have the same detailed data set.
    \item Both queries have the same dimensions and aggregate measure with a distributive function, $sum$.
    \item Both queries have selection conditions as conjunctions of one atom per dimension of the form $L \in \ \{v_1 \ldots v_k\}$.
    \item Both queries have perfectly rollable schemata with respect to their dimensions, with selection condition atoms being defined at ancestor levels of the respective groupers
    \item All levels of $q3$ are ancestors of the respective levels of $q2$
    \item Once we convert $\phi_3$ at the level of the schema of $q_2$ (grey band at the table in the lower part of the figure), the produced selection condition $\phi_3@2$ has grouper domains that are subsets of the respective grouper domains of $q2$
\end{enumerate}

To compute the new cube from the previous one, one has simply (a) to apply $\phi_3@2$ to $q2.cells$, and (b) to group the cells surviving the filter, according to their ancestors in $q3^+$. The facilitator of $sum$ is also $sum$, as far as the aggregate measure is concerned.

\begin{figure}[tb]
  \centering
    \includegraphics[width=0.9\textwidth]{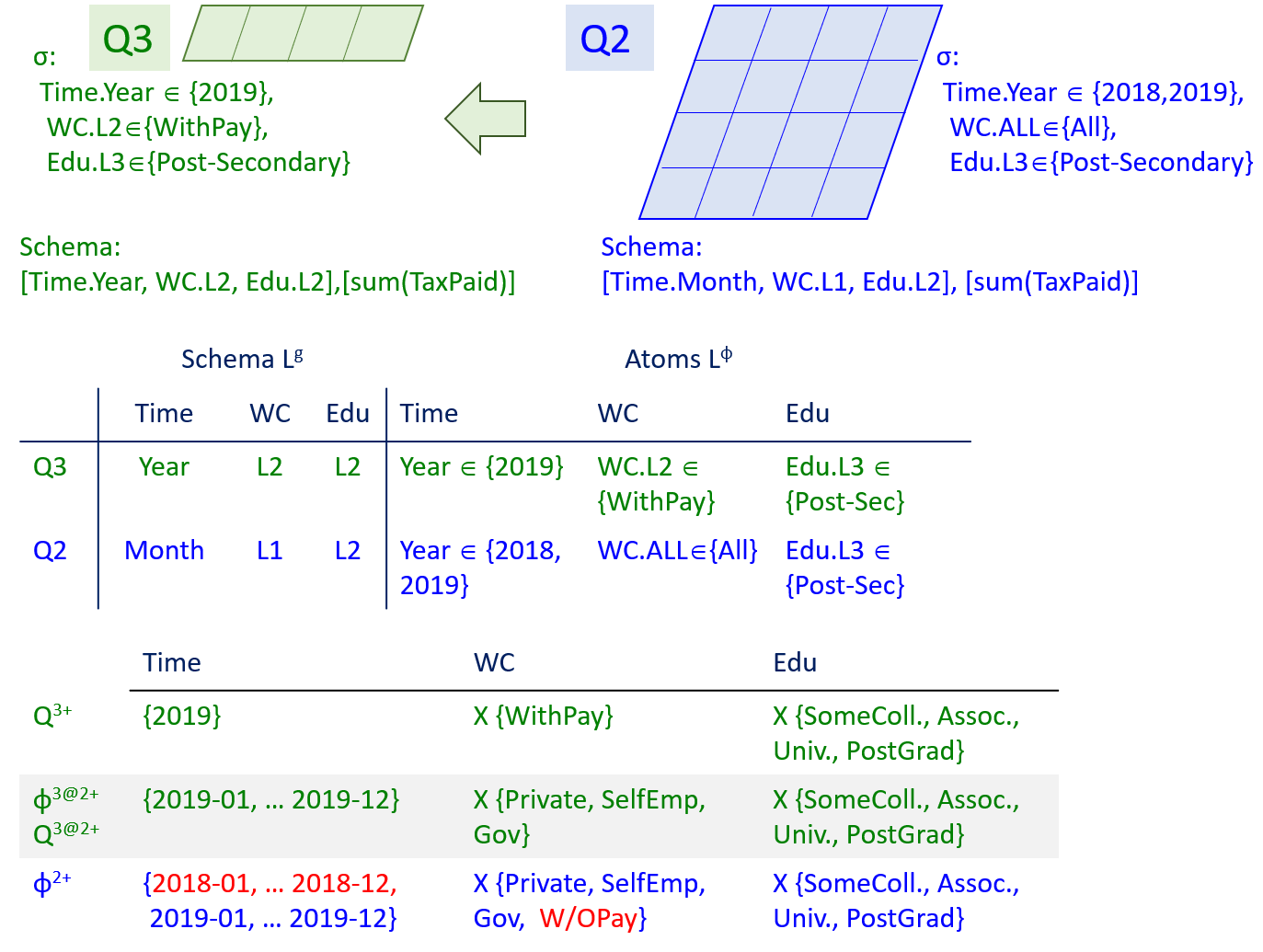}
      \caption{An example of cube usability}
      \label{fig:exUsable}
\end{figure}

\subsubsection{Comments and remarks}
There are several comments on variations of the usability theorem that can be done. We follow with the most basic ones.

\hrulefill\newline

\textbf{Measure incompatibilities}. A case we have not covered has to do with the possibility of the previous cube having more measures than the new one (the inverse case clearly collapses). The only requirement is to be able to derive a 1:1 mapping between the measures of the new and the previous query (which is done via the aggregation expressions, of course, and not via their names).

\hrulefill\newline

\textbf{Selections lower than groupers}. A case not covered by the previous proof has to do with the case where the selection condition of one of the queries has even one atom expressed at lower level than the grouper level. This makes the handling of such cases hard.

The salvation here comes from the perfect rollability property: since both queries come with a perfect rollability, this means that the low-level atom of the selection condition can be replaced by an absolutely equivalent atom at a higher level that is defined at a level at least as high as the one of the schema.
The automation of this task is left as an exercise for the reader.

\subsubsection{Easy Variants}
There are cases of very common operations in OLAP sessions where usability is inherently guaranteed. In other words, there are cases where there is no need to perform any check to decide whether the previous cube is usable or not.

\textbf{Roll-up}. When the selection conditions are identical and a roll-up is performed, if the new level of the roll-up is still equal or lower than the level of the selection condition, and the aggregate functions are all distributive, then all the properties are held by default, without any need for further checks. The only concern has to do with the case where the new query rolls-up higher than the selection condition. In this case, perfect rollability has to be checked.

\textbf{Extra filter}. Applying a new atom (or a conjunction of new atoms) to the previous cube, can potentially be handled directly, without resorting to the detailed cube. This can happen when the new atom involves the same level and a subset of the values of the previous query. To the extent that the atom of the new query is by definition expressed in a level higher than the one of its schema, and thus of the schema of the old query, (a) rollability is not affected, due to the height of the involved levels, and, (b) it is possible to transform the new atom over the schema (and thus the cell coordinates) of the previous cube. Practically, this means that we can immediately proceed in filtering out the unnecessary cells from the result of the new cube. 
The application of the new condition over the cells of the previous cube is straightforward: we compute the descendants of the values of the previous atom that are filtered out, and remove  from the result of the new query the cells having them as coordinates.

\textbf{Not applicable}. When applying drill down or drill across operations, it is impossible to exploit the previous cube; this, requires the computation of the query result from the underlying data set.

\silence{
\subsubsection{The absorption principle in logic}
The general idea of absorption in logic is:\\

If $\phi$ $\models$ $y$, then $\phi \wedge y$ $\models$ $\phi$.

Assume now that $\phi^{new}$ $\models$ $\phi^{old}$, in other words, the area of the multidimensional space that is covered by that $\phi^{new}$ is a subset of the respective area covered by that $\phi^{old}$. 

Then, ...
\subsubsection{Rollability inference for different levels of granularity}
Then, ...
}
\newpage
\section{Conclusions}\label{sec:concl}
In this paper we have provided a comprehensive model for multidimensional hierarchical spaces, selection conditions, cubes and cube queries, that can facilitate cube queries as well as the typical OLAP operators, as well as algorithms and checks for tasks involving comparing or computing cubes. We base the querying model on the hierarchical nature of the dimensions of data and require that all query semantics are defined with respect to the most detailed level of aggregation in the hierarchical space. This allows us to define equivalent expressions at different levels of granularity. Based on these premises, we provide checks and algorithms for deciding and computing results for the following problems:
\begin{itemize}
    \item Foundational containment, i.e.,  whether the area that pertains to a cube $c^n$ is a subset of the respective area that pertains to another cube $c^b$.
    \item Same-level containment, where we want to determine whether and how one cube's cells are a subset of the cells of another cube, based only on the syntactic expression, for cubes that are defined at the same level of abstraction with respect to their dimensions.
    \item Cube-intersection, i.e., deciding whether, and to what extent, the results of two cube queries overlap
    \item Query distance, i.e., given the syntactic definition of two cubes, being able to assess how similar they are.
    \item Usability, i.e., the possibility of computing a new cube from a previous one, defined at a different level of abstraction. 
\end{itemize}

Future work can primarily target more sophisticated operators. This can include comparing cubes for intrinsic properties of their cells (e.g., hidden correlations, predictions, classifications) that have to be decided via the application of knowledge extraction operators to the results, or the detailed areas, of the contrasted cubes. 
\newpage
\bibliographystyle{alpha}
\bibliography{biblio, cubeJ}
\end{document}